\documentclass[article, amsmath,amssymb,
 aps, physrev]{revtex4-2}
\usepackage{amsmath,amsthm}
\usepackage{xcolor}
\usepackage{natbib} 
\usepackage{graphicx}
\usepackage{subfig}
\usepackage{dcolumn}
\usepackage{bm}
\usepackage{comment}
\usepackage[justification = raggedright]{caption} 

\newtheorem{proposition}{Proposition}
\newtheorem{lemma}{Lemma}
\newsavebox{\tempfig}

\begin{document}



\title{The `Brazil-nut effect' in bidisperse particle laden flow on an incline} 

\author{Jack Luong}
 \email{Contact author: jqluong@math.ucla.edu}
\affiliation{%
\\
 Department of Mathematics, University of California Los Angeles, Los Angeles 90095, CA, United States of America
}%
\author{Sarah Cassie Burnett}%
 \email{Contact author: burnett@math.ucla.edu}
\affiliation{%
\\
 Department of Mathematics, University of California Los Angeles, Los Angeles 90095, CA, United States of America
}%
\author{Andrea L. Bertozzi}
 \email{Contact author: bertozzi@math.ucla.edu}
\affiliation{%
 Department of Mathematics and Mechanical and Aerospace Engineering, University of California Los Angeles, Los Angeles 90095, CA, United States of America
}%

\begin{abstract}
 We study bidisperse suspensions---suspensions where there are two particle species of the same density but different sizes---of a viscous fluid on an incline. We 
 use a lubrication theory/thin film model
 to form a hyperbolic system of three conservation laws for the height and particle volume fractions. 
 The model predicts, over a range of parameters, that the larger particles rise to the top of the layer, consistent with the well-known `Brazil-nut effect' for granular media. The model predicts well-separated fronts of the two species of particles, behind a clear fluid front, at lower inclination angles and volume fractions. This corresponds to a triple shock structure in the system of conservations. At higher inclination angles and volume fractions the particles congregate at a high concentration at the leading front corresponding to a singular shock in the model. We find excellent agreement between theory and experiments in terms of the overall dynamic structures as the parameters vary.
 
\end{abstract}

\maketitle
\section{Introduction}\label{sec:Introduction}
Particle laden flows, mixtures of particles and fluid under the influence of forces, arise in many different manufacturing, environmental, and industrial applications. Granular flows on an incline are well studied for their impact on debris flow in large-scale phenomena such as avalanches and landslides \cite{santangelo2021new} and are also modeled at a fundamental level, with theoretical interest in their rheology \cite{guazzelli2024rheology} and particle segregation \cite{larcher2019influence}. The study of particle laden flows are also of interest to the food industry with the Bostwick consistometer \cite{balmforth_viscoplastic_2007, mouquet_characterization_2006, tehrani_modification_2007}, a device used to determine the viscosity of food products. Spiral separators are used in the mining industry to separate wet mixtures of particles based on density and/or particle size \cite{WrightpatentEU, dehaine_modelling_2016, holland-batt_particle_1991}. 

In this work, we consider particles in a viscous layer of fluid on an incline.
In particular, we examine a slurry with two negatively buoyant particle species of the same density but different sizes. The particles are sufficiently large to avoid being influenced primarily by Brownian motion, yet small enough to stay suspended and maintain the continuum nature of the fluid. We are interested in understanding how the particles naturally segregate by size over time as they move down the inclined geometry.
It is well-known for dry, granular flows that many-body interactions between different-sized particles, under gravity, often result in the larger particles rising to the top as demonstrated in Fig.~\ref{fig:brazil}. This effect is colloquially known as the `Brazil-nut effect'. Investigation into the Brazil-nut effect began with the seminal work of Rosato et.~al in \cite{rosato_why_1987} where the authors attribute the Brazil-nut effect to the ability of smaller particles to fit into the voids surrounding larger particles---spaces that larger particles themselves cannot occupy. The authors of \cite{jullien_mechanism_1990} further this argument by considering a three dimensional geometry and arguing smaller particles roll over larger ones. Many other mechanisms have been proposed since then. The authors in \cite{knight_vibration-induced_1993} present particles shaken in a cylinder follow a convective cycle up the middle of the container and then travel down the walls. The larger particles are trapped when they reach the top of the convection zone, while the smaller ones are able to convect freely. A different study \cite{mobius_size_2001} indicates particle density as a factor affecting sedimentation. Another proposed mechanism is granular temperature, the velocity variance, as the segregating mechanism that behaves analogously to thermal temperature as discussed in \cite{fan_phase_2011}. Numerous experimental work such as \cite{knight_experimental_1996} and \cite{gajjar_size_2021} support these investigations into the Brazil-nut effect. 
\begin{figure}
 \subfloat[Before shaking]{\label{subfig:brazil_before}
 \includegraphics[width=0.35\linewidth]{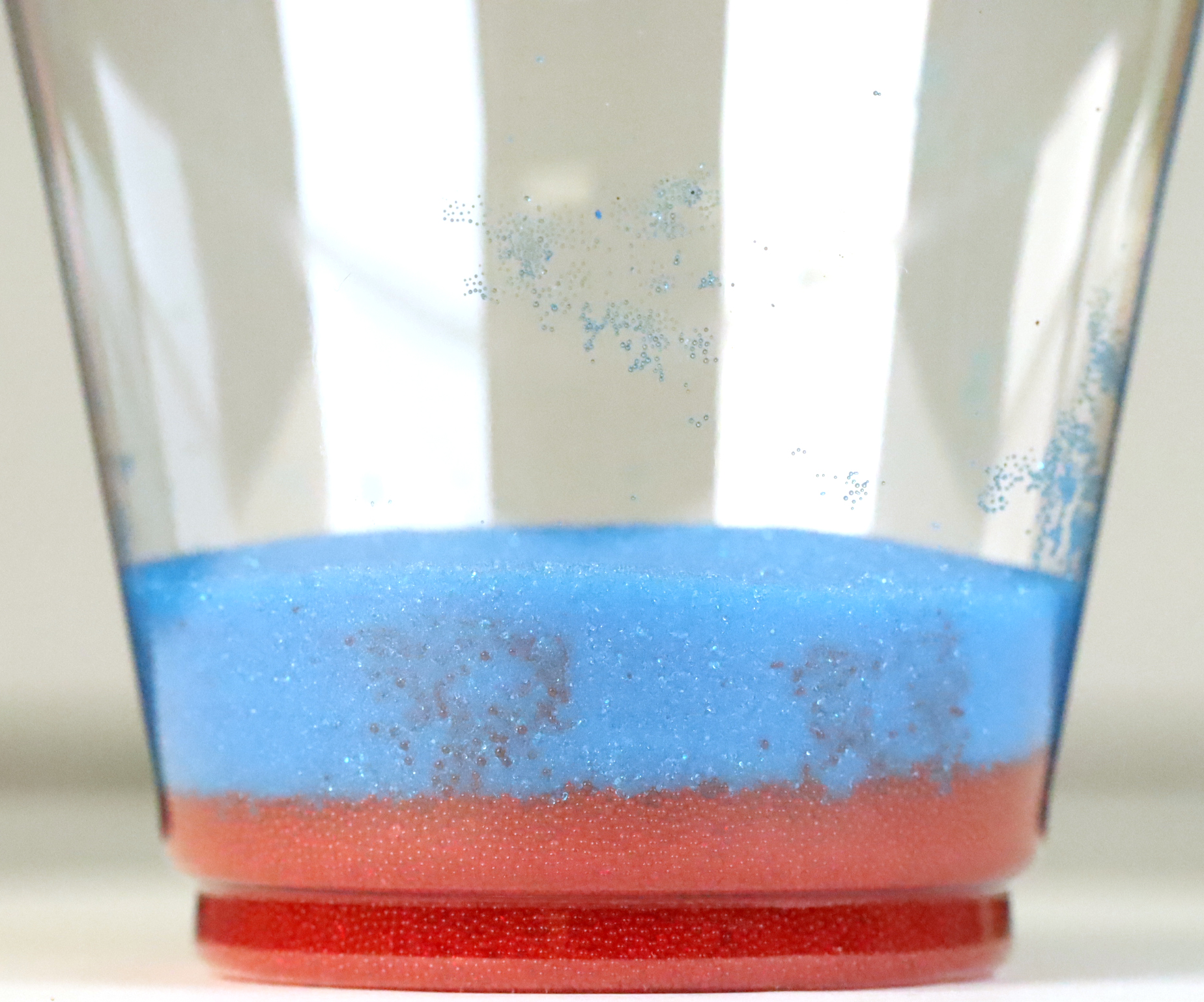}} ~
 \subfloat[After shaking]{\label{subfig:brazil_after}
 \includegraphics[width=0.35\linewidth]{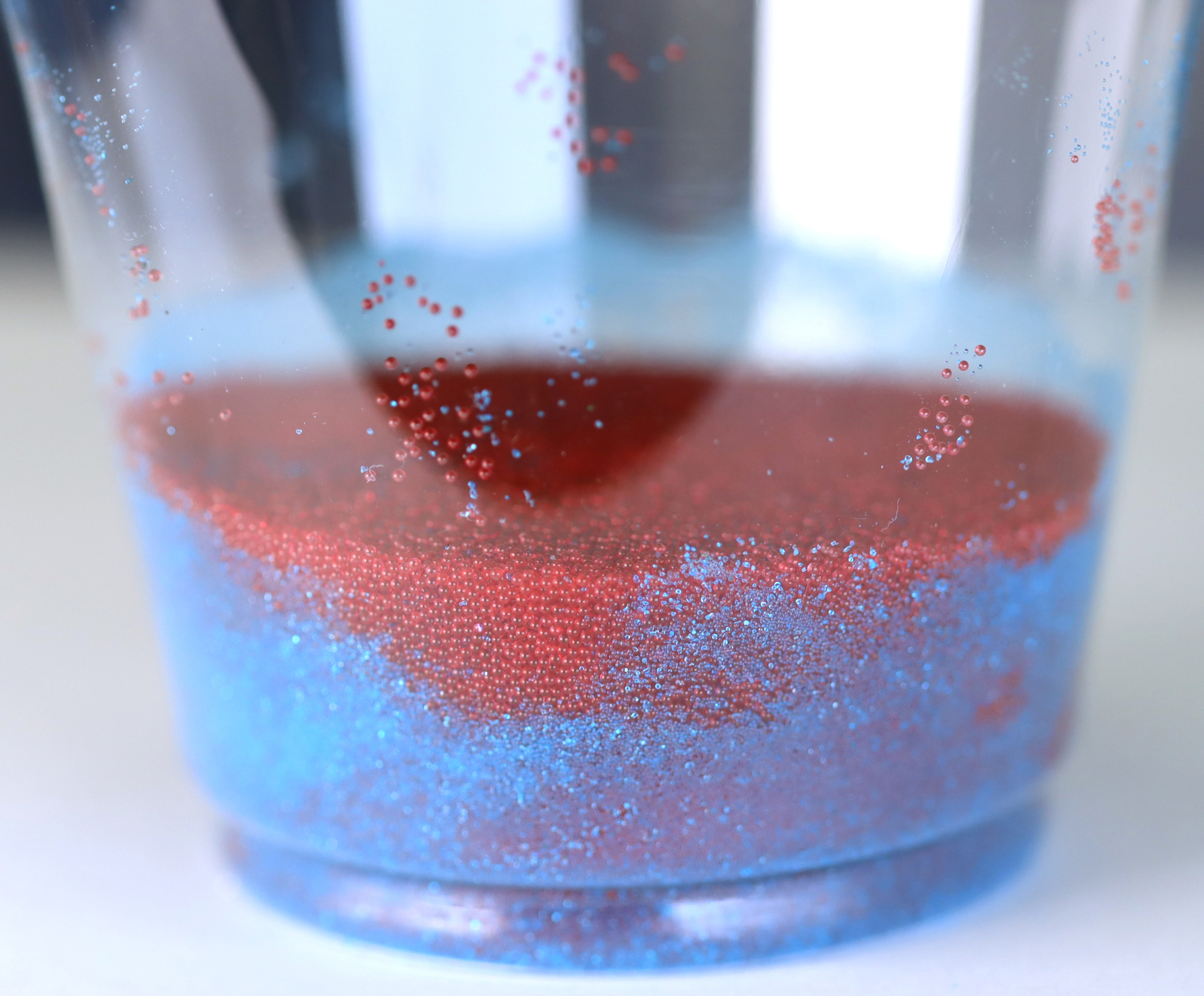}} \\
 \subfloat[Before settling]{\label{subfig:stokes_before}
 \includegraphics[width=0.35\linewidth]{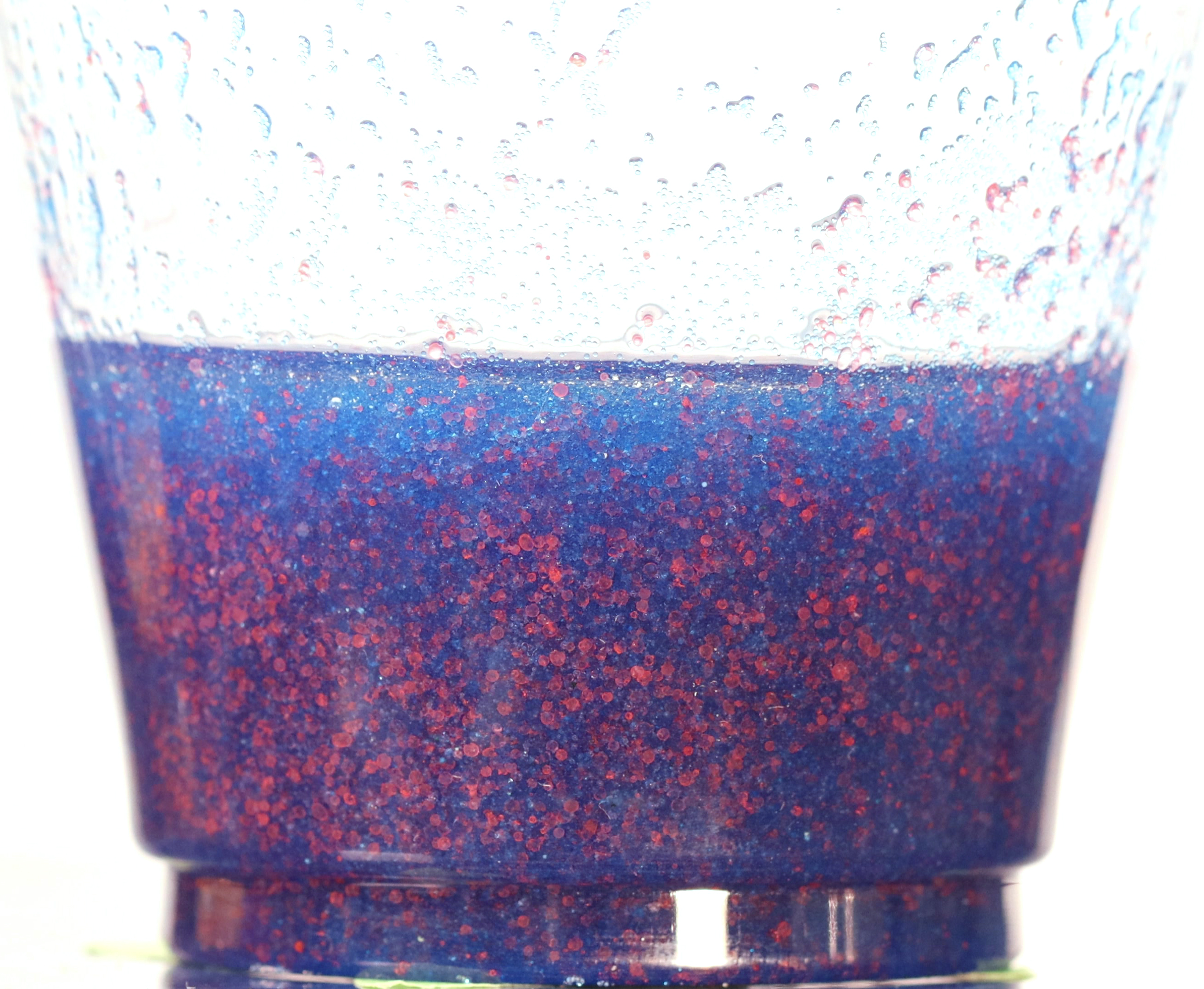}} ~
 \subfloat[After settling]{\label{subfig:stokes_after}
 \includegraphics[width=0.35\linewidth]{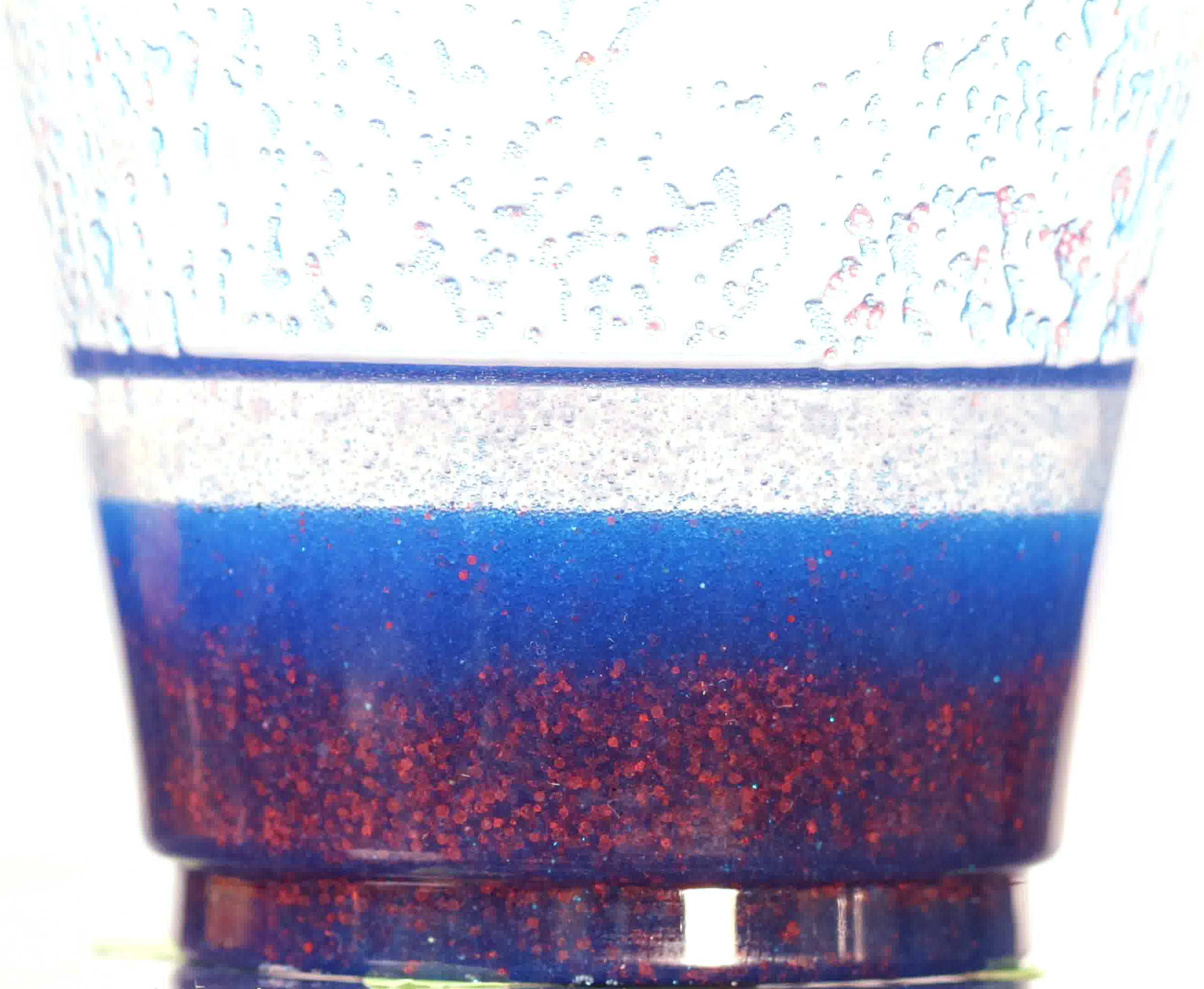}}
\caption{Demonstration of different effects achievable in bidisperse mixtures of particles.  In Fig.~\ref{subfig:brazil_before}, equal volumes of red and blue glass particles are added to a cup, with the blue particles layered above the red ones. The red glass particles are larger and have $0.6$ mm diameter. The blue glass particles are smaller and have $0.3$mm diameter. When vigorously shaken in all directions for 3 minutes, the larger particles rise to the top as shown in Fig.~\ref{subfig:brazil_after}---an example of the Brazil-nut effect. On the other hand, when the same particles are suspended in a viscous fluid, as in Fig.~\ref{subfig:stokes_before}, mixed well, and then allowed to settle, the larger particles settle at a faster rate than the smaller ones. After 15 minutes of settling, there's a distinct layer of blue smaller particles on top in Fig.~\ref{subfig:stokes_after} \cite{burnett2024gfm}.}
\label{fig:brazil}
\end{figure}
\par The Brazil-nut effect has also been studied in granular flows on an incline. The authors of \cite{savage_particle_1988} attribute the Brazil-nut effect to the same void filling mechanism as in \cite{rosato_why_1987} and contact forces squeezing particles out of their layer. In a high speed setting, the authors of \cite{neveu_particle_2022} use gradients of granular temperature to explain particle migration. Understanding of the Brazil-nut effect is of interest to the mining and pharmaceutical industry as the tendency for large particles to aggregate at the top prevents well mixing \cite{johanson_j_r_particle_1978}.

\par We contrast this phenomena with Stokes' law, where larger spheres in a viscous fluid moving at a very low particle Reynolds number ($\text{Re} = \rho_\ell d^2 \dot{\gamma} / (4\mu_\ell) $) \cite{murisic_dynamics_2013} move faster than smaller spheres as seen in Figs.~\ref{subfig:stokes_before} and \ref{subfig:stokes_after}. We find in bidisperse particle laden flow the larger particles rise to the top like in granular flows and unlike in Stokes flow, as seen in Fig.~\ref{subfig:sideprofile}. Here we develop a continuum model for the dynamic motion of particles (according to volume fraction) that explains why we observe the Brazil-nut effect in our bidisperse slurry flow.

\begin{figure}[h]
 \centering
 \subfloat[Side Profile\label{subfig:sideprofile}]{\includegraphics[width=0.7\linewidth]{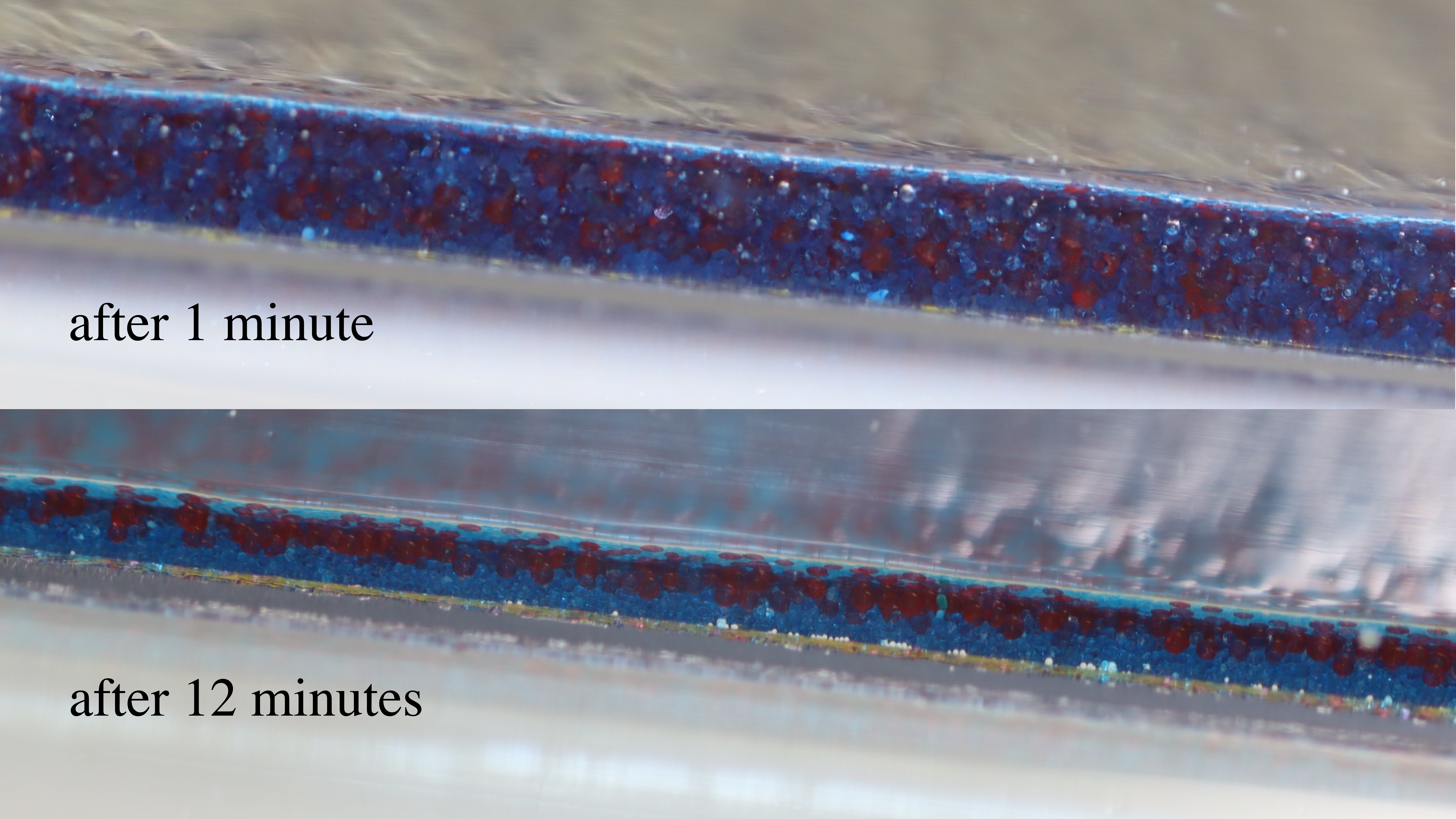}} \hspace{1em}
 \subfloat[Aerial Profile\label{subfig:settled}]{\includegraphics[width=0.25\linewidth]{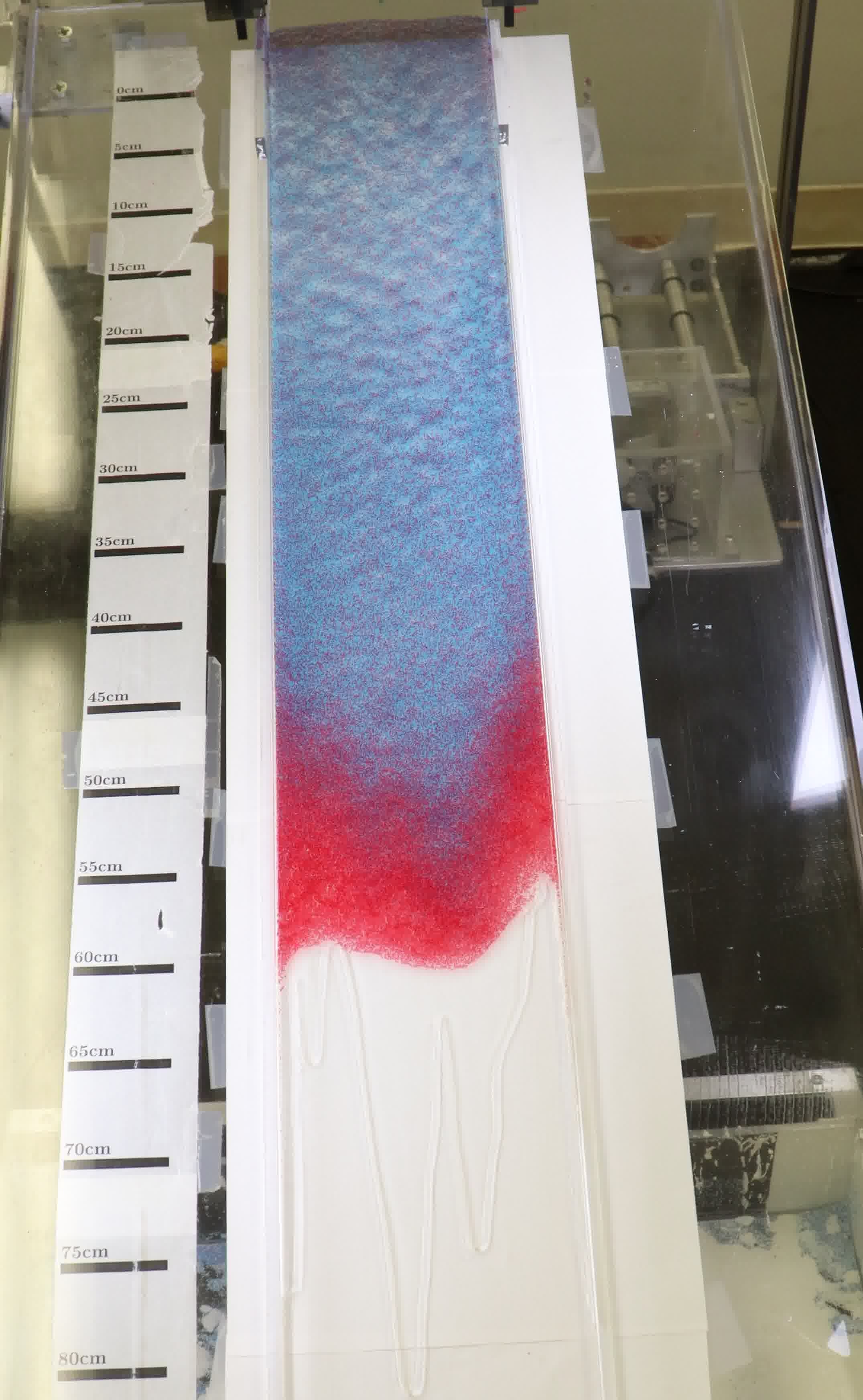}}
 \caption{Side profile of an experiment in the settled regime and the bird's eye perspective of the flow. The red particles are the larger particles.}
 \label{fig:sideprofile}
\end{figure}

Previous work on thin-film slurries has built on a classical model for 
clear fluid flow on an incline, studied by 
Huppert in \cite{huppert1982flow}.
He used a conservation law for the flow, arising from lubrication theory, and derived analytic, self-similar, formulas for the front speed and fluid height for the finite volume problem. The authors of \cite{zhou2005theory} then investigated the same geometry with a fluid containing one species of non-neutrally buoyant particles of identical size and density. In that work they included hindered settling of particles in the model. Cook \cite{cook2008theory} further advanced the model by including shear-induced migration, which quantitatively explained a natural bifurcation in the experiment in which the particles preferentially settle to the substrate or congregate at the leading front depending on the volume fraction and inclination angle. This model was quantitatively validated by detailed experiments in \cite{murisic_particle-laden_2011}. The work \cite{murisic_dynamics_2013} develops a quantitative dynamic lubrication model that shows excellent quantitative agreement with dynamic experiments in the settled regime.
Using a similar framework, \cite{wong_conservation_2016} extended the model in \cite{murisic_dynamics_2013} to the case of two particle species of the same size but different densities.
That work followed the equilibrium theory derivation in \cite{lee_equilibrium_2015} to derive the dynamic conservation laws. As in the earlier works, the dynamic model assumes rapid equilibration of the stratification of particles in the normal direction to the substrate, compared with the speed of the front. 
 
Other authors have studied bidisperse particle laden flow of two different sizes in other geometries. The authors of \cite{shauly_shear-induced_1998} model shear-induced migration of polydisperse suspensions and introduce an dynamic expression for the maximum packing fraction in the bidisperse case depending on the relative diameters of the particles and their volume fractions. In \cite{kanehl_hydrodynamic_2015}, the authors examine the case of bidisperse suspensions in channel flow and expand upon the modeling of the shear-induced migration flux term for the bidisperse case. In both of these works the authors find the particles segregate within the flow by size. The authors of \cite{howard_bidisperse_2022} approach the modeling of neutrally buoyant bidisperse particle laden flow in a channel using the suspension balance model. Since we consider negatively buoyant bidisperse particle laden flow down a slope, we include gravitational settling in our analysis which leads to more varied qualitative outcomes. In contrast, the authors of \cite{thornton_three-phase_2006} use the modeling techniques established in the granular convection literature study bidisperse particle laden flow with a generic interstitial fluid. However, they make the simplifying assumption that the particle volume fraction is constant. In the more complex spiral geometry, particle species separate by density through a similar thin-film model, with long-time behavior analyzed once the cross-sectional flow has equilibrated \cite{ding_equilibrium_2025}.

This paper is organized as follows.
In  Sec.~\ref{sec:model}, we derive the conservation law model that governs how the slurry develops as it flows down the incline. We also present the equilibrium model that dictates how the particle species and fluids equilibriate in the direction normal to the incline. In  Sec.~\ref{sec:analysis}, we analyze the bidisperse equilibrium model which yields the counterintuitive result of coarser particles accumulating near the surface and finer particles settling at the bottom. In  Sec.~\ref{sec:experiment}, we go over the experimental methods and materials. In  Sec.~\ref{sec:discussion}, we qualitatively compare numerical simulations of our dynamic model to laboratory experiments, highlighting the similar features we observe in both. In  Sec.~\ref{sec:conclusion}, we finish with concluding remarks and potential future directions.

\section{Model}\label{sec:model}
In this section, we develop a system of conservation laws in the continuum limit that describe the evolution of the film height and particle concentration in the slurry as it slides down the incline. Figure \ref{fig:diagram} shows the geometry of the problem.
\begin{figure}
 \centering
 \includegraphics[width=0.5\linewidth]{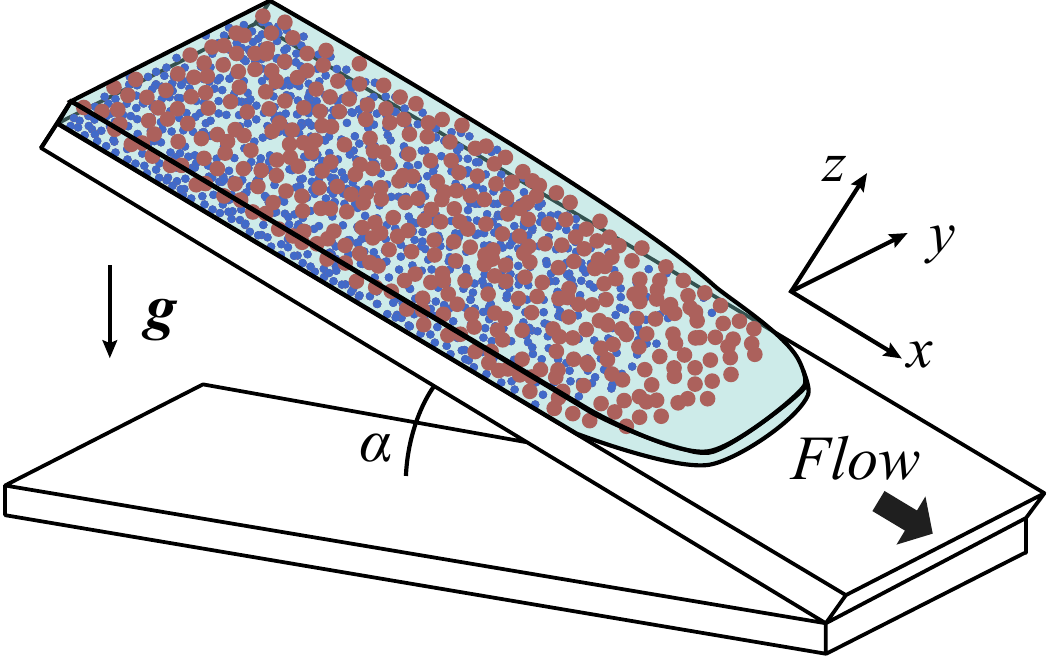}
 \caption{Schematic of bidisperse particle laden flow developing on an incline}
 \label{fig:diagram}
\end{figure}
The slurry is comprised of a fluid with viscosity $\mu_\ell$ and density $\rho_\ell$, and two different species of negatively buoyant particles with the same density $\rho_s$ and different diameters $d_1 > d_2$. This slurry is confined to a track of fixed width and slides down an incline of angle $\alpha$. The viscosity of the slurry is a function of the particle concentration, $\phi$, and is given via the Krieger-Doughertry relation as $\mu(\phi) = \mu_\ell(1 - \phi / \phi_m)^{-2}$ \cite{krieger_mechanism_1959} where $\phi_m$ is the maximum packing fraction of spheres in the fluid. Unlike the case with particles of the same size and different densities, $\phi_m$ depends on the particle concentration of particles of different sizes. This is due to the fact smaller particles can fit into the gaps between larger particles that larger particles cannot. We denote $\phi_1$ as the particle concentration of larger particles and $\phi_2$ as the particle concentration of smaller particles within the slurry, with the relation $\phi = \phi_1 + \phi_2$. The dependence of $\phi_m$ on $\chi = \frac{\phi_1}{\phi}$, the proportion of the larger particles to all of the particles in the slurry, is given as 
\begin{equation}\label{eq:phimax_sec:model}
 \phi_m = \phi_{m,0}\left(1 + \frac{3}{2} b^{\frac{3}{2}} \chi^{\frac{3}{2}} \left(1 - \chi\right)\right)
\end{equation}
where $b = \frac{d_1 - d_2}{d_1 + d_2}$ and $\phi_{m,0}$ is the maximum packing fraction of spherical particles of the same size in the fluid \cite{shauly_shear-induced_1998}. We choose $\phi_{m,0} = 0.61$ \cite{ward_experimental_2009}.
\par We now state the conservation of momentum and conservation of mass equation for the slurry respectively:
\begin{equation}\label{eq:momentum_conservation_sec:model}
 \rho\left(\partial_t \Vec{u} + \Vec{u} \cdot \nabla \Vec{u}) = \nabla \cdot \left(-p I + \mu(\nabla \Vec{u} + \nabla \Vec{u}^T\right)\right) + \rho \Vec{g}
\end{equation}
and
\begin{equation}\label{eq:mass_conservation_sec:model}
 \partial_t \rho + \nabla \cdot (\rho \Vec{u}) = 0
\end{equation}
where $\Vec{u} = (u,w)$ represents the $x$ and $z$ components of velocity respectively and $\Vec{g} = (0,-g)$ where $g$ is gravitational acceleration. Equation \eqref{eq:momentum_conservation_sec:model} describes how momentum of the fluid is conserved and equation \eqref{eq:mass_conservation_sec:model} describes how mass of the fluid is conserved. However, we assume the fluid has low Reynolds number, so the momentum terms---the left hand side of Eq.~\eqref{eq:momentum_conservation_sec:model}---can be neglected. The velocity satisfies the no-slip boundary condition ($u = 0$ when $z = 0$) and no stress ($\Vec{n} \cdot \left(-p I + \mu(\nabla \Vec{u} + \nabla \Vec{u}^T\right) = 0 $) occurs at the free surface. We also account for the conservation of each species of particles with
\begin{equation}\label{eq:particle_conservation_sec:model}
 \phi_{i,t} + \Vec{u} \cdot \nabla \phi_{i} + \nabla \cdot \Vec{J} = 0 \quad i = 1,2
\end{equation}
where $\Vec{J}$ represents all of the particle flux terms.
\par We now nondimensionalize the system using a thin film approximation as done in \cite{lee_equilibrium_2015} under the following scaling:
\begin{eqnarray}\label{eq:dimensionless}
 (\hat{x},\hat{z}) = \frac{1}{H}(\delta x, z), \quad
 \hat{u} = \frac{1}{U_0} \left(u, \frac{w}{\delta}\right), \quad
 \hat{J} = \frac{H^2}{\Bar{d}^2U_0}\left(\frac{J_x}{\delta}, J_z\right), \nonumber \\
 \hat{p} = \frac{H}{\mu_\ell U_0}, \quad
 \hat{\mu} = \frac{\mu}{\mu_\ell}, \quad 
 \hat{\rho_s} = \frac{\rho_s - \rho_\ell}{\rho_\ell}, \quad
 s = \frac{z}{H}, \nonumber 
\end{eqnarray}
where $H$ and $L$ are the characteristic film thickness and the axial length scale respectively, $U_0 = H^2 \rho_\ell g \sin{\alpha} / \mu_\ell$, $\delta = \frac{H}{L} << 1$, and $\Bar{d}$ is the average diameter of the particles. Going forward the hats will be dropped for brevity.

\par Next, we make the assumption, as done in previous work \cite{murisic_dynamics_2013, murisic_particle-laden_2011, wong_conservation_2016, ding_equilibrium_2025}, the particles equilibriate in the normal direction faster than fluid flows down the incline. We can conclude this by studying the leading order behavior of Eqs. \eqref{eq:momentum_conservation_sec:model} and \eqref{eq:mass_conservation_sec:model} alongside the assumption that $\delta << (\bar{d}/H)^2 << 1$. The upper bound ensures the particles are not too large and is consistent with the continuum model. The lower bound ensures the particles are not too small and their motion is not dominated by Brownian motion \cite{ding_equilibrium_2025}.

\par First, we study the $x$ component of the momentum equation \eqref{eq:momentum_conservation_sec:model}, which reduces to 
\begin{equation}\label{eq:velocity_ODE}
 \frac{\partial}{\partial z}\left(\mu  \frac{\partial u}{\partial z}\right) = -\rho
\end{equation} 
in leading order of $\delta$. The no-slip boundary condition on the surface of the incline gives us $\partial_z u|_{z = 0} = 0$ and no stress at the free surface of the slurry gives us $\mu  \partial_z u |_{z = h} = 0$. This leads to the following ODE in terms of the stress $\sigma$ in the $z$ direction:
\begin{equation}\label{eq:sigma'}
\frac{\partial \sigma}{\partial z} = -1 - \rho_s \phi, \quad
\sigma(0) = 1 + \rho_s \int_0^z \phi(s) ds, \quad
\sigma(h) = 0.
\end{equation}
Similar analysis for the conservation of particles yields $\partial_z J_z = 0$ for both species. As particles cannot exit the slurry through the track or the free surface, this implies
\[
J_{z,i} = 0,
\]
meaning each species of particles is in equilibrium in the $z$ direction.

\par To describe how the system evolves in the $x$ direction, we need to know the velocity in the $x$ direction, $u$. From Eqs.~\eqref{eq:velocity_ODE} and \eqref{eq:sigma'}, we obtain the ODE
\begin{equation}\label{eq:velocity}
 \sigma = \mu(\phi) \frac{\partial u}{\partial z}.
\end{equation}
In order to compute $u$, we need to first compute the viscosity and the shear stress, which requires knowledge of the particle concentration of both species in the $z$ direction.

\par We now examine the flux balance of particles in the $z$ direction by describing each of the individual components of the particle flux. We claim the total flux for each particle species $i$ is comprised of three components: shear-induced migration, settling also known as sedimentation, and self-mixing between particle species yielding 
\begin{equation}\label{eq:total_flux}
 J_i = J_{\text{shear,i}} + J_{\text{settling,i}} + J_{\text{tracer,i}}.
\end{equation}
We have dropped the $z$ subscript from the flux terms $J$ for brevity. We now define each flux term in detail. 
\par Shear-induced migration represents the migration of particles from high to low concentrations and from high to low shear stress \cite{leighton1987shear, phillips1992constitutive}. The shear-induced migration term largely remains the same from its description in \cite{murisic_particle-laden_2011} and \cite{wong_conservation_2016}, but the different diameters of the two particle species leads to different displacements when collision occurs. The shear-induced migration term is thus

\begin{equation} \label{eq:J_shear}
 J_{\text{shear,i}} = -\phi_i \sum_{j=1}^2 A_{ij} \left( K_c \nabla(\phi_j \dot{\gamma}) + K_v \frac{\phi_j \dot{\gamma}}{\mu(\phi)} \nabla(\mu) \right)
\end{equation}
where $K_c \approx 0.41$ and $K_v \approx 0.62$ are empirically determined constants \cite{phillips1992constitutive}, the coupling matrix $A$ is
\begin{equation}\label{eq:A_matrix}
 A_{ij} = \frac{1}{4} \frac{(d_i + d_j)^2}{2^{d+1}} \frac{(1 + \frac{d_i}{d_j})^d}{1 + (\frac{d_i}{d_j})^d},
\end{equation}
$\dot{\gamma} = \sigma / \mu$is the shear rate for a Newtonian fluid, and $d$ is the dimension of the spheres (in this case, $d = 3$). We note that the diagonal terms of this matrix reduces to $d_1^2/4$ and $d_2^2/4$. Also, $A$ is symmetric as the off diagonal terms are identical:
\[
A_{ij} = \frac{1}{4} \frac{(d_i + d_j)^2}{2^{d+1}} \frac{(1 + \frac{d_i}{d_j})^d}{1 + (\frac{d_i}{d_j})^d} \cdot \frac{(\frac{d_j}{d_i})^d}{(\frac{d_j}{d_i})^d} = \frac{1}{4} \frac{(d_i + d_j)^2}{2^{d+1}} \frac{(\frac{d_j}{d_i} + 1)^d}{(\frac{d_j}{d_i})^d + 1} = A_{ji}.
\]
Equations \eqref{eq:J_shear} and \eqref{eq:A_matrix} are adapted from \cite{kanehl_hydrodynamic_2015}. In shear-induced migration, the collisions from the particles as they migrate cause different displacements depending on the pair of particle species. Thus, the shear-induced migration of species $i$ has a different contribution from the two particle species, and the total flux is the sum of the contributions from each particle species. 

\par Now we move on to the settling flux. Since the particles are of the same density, we do not need to account for special particle interactions between species unlike in \cite{tripathi_viscous_1999}. We also assume particles settle the same regardless of their size. This leads us to the settling flux as originally described in \cite{leighton_measurement_1987} with hindrance function $f(\phi) = \mu_\ell (1-\phi)/\mu(\phi)$:
\begin{equation}\label{eq:J_settling}
 J_{\text{settling},i} = -\frac{2 d_i^2 \phi_i \cot{\alpha}(1 - \phi) (\rho_i - \rho_\ell) }{9 \mu (\phi)}.
\end{equation}

\par Finally, we have the tracer flux which captures the random motion of particles in a sheared mixture \cite{leighton_measurement_1987}. While the net concentration of particles $\phi$ does not change under tracer flux, the individual particle concentrations $\phi_i$ do. The tracer flux for a particle species is given as 
\begin{equation}\label{eq:J_tracer}
 J_{\text{tracer},i} = -\frac{\dot{\gamma}d_i^2}{4} D_{\text{tr}}(\phi) \phi \nabla\left(\frac{\phi_i}{\phi}\right)
\end{equation}
where $ D_{\text{tr}}$ is the tracer diffusivity and is given by the empirical expression $ D_{\text{tr}}(\phi) = \frac{1}{2}\min \{\phi^2, \phi_{\text{tr} }^2 \}$ \cite{lee_equilibrium_2015}. In the dilute limit, \cite{leighton1987shear, leighton_measurement_1987} proposes the empirical expression $\phi^2/2$ for tracer diffusivity, while \cite{sierou_shear-induced_2004} notes the tracer diffusivity becomes a constant $\phi_{\text{tr}} = 0.4$ for large concentrations.

\par Now, we express the flux balance in Eq.~\eqref{eq:total_flux} for each particle species as an ODE in terms of the particle concentration $\phi$ and the species ratio $\chi$. Writing this in terms of $\phi'$ and $\chi'$ where the derivatives here are with respect to $z$, we obtain
\begin{equation}\label{eq:ODE_full}
\begin{split}
 -\begin{bmatrix} 
 \phi_1 & 0 \\
 0 & \phi_2
 \end{bmatrix}
 A
 \begin{bmatrix}
 \chi g & h_1 \\ (1-\chi)g & h_2
 \end{bmatrix}
 \begin{bmatrix}
 \phi' \\ \chi'
 \end{bmatrix} - 
 \frac{\sigma}{4} \phi D_{tr}(\phi) \begin{bmatrix}
 d_1^2/4 & 0 \\ 0 & d_2^2/4
 \end{bmatrix}
 \begin{bmatrix}
 0 & 1 \\ 0 & -1
 \end{bmatrix}
 \begin{bmatrix}
 \phi' \\ \chi'
 \end{bmatrix}
 = \\
 \frac{2\cot{\alpha}}{9}\rho_s(1-\phi)\begin{bmatrix}
 \phi_1 & 0 \\
 0 & \phi_2
 \end{bmatrix}\begin{bmatrix}
 d_1^2/4 & 0 \\ 0 & d_2^2/4
 \end{bmatrix} \begin{bmatrix}
 1 \\ 1
 \end{bmatrix} 
 + K_c \begin{bmatrix}
 \phi_1 & 0 \\
 0 & \phi_2
 \end{bmatrix}
 A\begin{bmatrix}
 \phi\chi\sigma' \\ \phi(1-\chi)\sigma'
 \end{bmatrix} \\
\end{split}
\end{equation}
where:
\begin{equation*}
    \begin{split}
        g   = K_c \left( \sigma - \phi\sigma \left( \frac{2}{\phi_m - \phi} \right) \right) 
       + K_v \phi \sigma \left( \frac{2}{\phi_m - \phi} \right), \quad 
h_1 = K_c \left( \phi \sigma + \phi \chi \sigma \left( \frac{2}{\phi_m - \phi} \right) \xi \right), \\
h_2 = K_c \left( -\phi \sigma + \phi (1 - \chi) \sigma \left( \frac{2}{\phi_m - \phi} \right) \xi \right), \quad 
\xi = \frac{\phi}{\phi_m} \left( \phi_m - \phi_{m,0} 
       \frac{3 - 5X}{2X(1 - X)} \right).
    \end{split}
\end{equation*}

\par The system in equation \eqref{eq:ODE_full} is linear in the variables $\phi'$ and $\chi'$, and recalling $\phi_1 = \phi \chi$ and $\phi_2 = \phi (1 - \chi)$, we obtain the following ODEs for $\chi'$ and $\phi'$:
\begin{multline}\label{eq:chi'}
 \chi' = \frac{2\cot{\alpha}}{9}\rho_s (1-\phi)(1-\chi)(\chi)\left[(2\chi -1)d_1^2d_2^2 + a(d_1^2 - (d_1^2+d_2^2)\chi \right] 
 \\ \times \left[ \sigma \left(K_c \det{A} (1-\chi)\phi \chi + D_{tr}(\phi) \left[d_1^2 d_2^2 (1-\chi)^2 +a(d_1^2 + d_2^2)(1-\chi)(\chi) + d_1^2d_2^2 \chi^2 \right] \right) \right]^{-1}
\end{multline}
and
\begin{align}\label{eq:phi'}
\begin{split}
\phi' &= 
- \left[ \left( \chi + \frac{a}{d_2^2}(1 - \chi) \right) h_1 
+ \left( (1 - \chi) + \frac{a}{d_1^2} \chi \right) h_2 \right] \chi' \\
&\quad + \frac{2 \cot{\alpha}}{9} \rho_s (1 - \phi) 
+ K_c \phi \sigma' \left( \chi^2 
+ a \chi (1 - \chi) \left( \frac{1}{d_1^2} + \frac{1}{d_2^2} \right) 
+ (1 - \chi)^2 \right) \\
&\quad \times \left[ 
g \left( \chi^2 
+ a \left( \frac{1}{d_1^2} + \frac{1}{d_2^2} \right) \chi (1 - \chi) 
+ (1 - \chi)^2 \right) 
\right]^{-1}, \\[1em]
\end{split}
\end{align}
where $a$ is the off diagonal entry of $A$ with $\frac{1}{4}$ factored out:
\begin{equation}\label{eq:A_offdiag}
 a = \frac{(d_1 + d_1)^2}{2^{d+1}} \frac{(d_1 + d_2)^d}{d_1^d + d_2^d}.
\end{equation}

Coupled with the ODE for stress given in Eq.~\eqref{eq:sigma'} and the appropriate boundary and integral conditions, we obtain the following system that describes how the particles equilibriate in the $z$ direction:

\begin{align}\label{eq:ODE_phi_chi_sigma}
\begin{split}
\phi' &= 
- \left[ \left( \chi + \frac{a}{d_2^2}(1 - \chi) \right) h_1 
+ \left( (1 - \chi) + \frac{a}{d_1^2} \chi \right) h_2 \right] \chi' \\
&\quad + \frac{2 \cot{\alpha}}{9} \rho_s (1 - \phi) 
+ K_c \phi \sigma' \left( \chi^2 
+ a \chi (1 - \chi) \left( \frac{1}{d_1^2} + \frac{1}{d_2^2} \right) 
+ (1 - \chi)^2 \right) \\
&\quad \times \left[ 
g \left( \chi^2 
+ a \left( \frac{1}{d_1^2} + \frac{1}{d_2^2} \right) \chi (1 - \chi) 
+ (1 - \chi)^2 \right) 
\right]^{-1}, \\[1em]
\chi' &= \frac{2\cot{\alpha}}{9}\rho_s (1-\phi)(1-\chi)(\chi)\left[(2\chi -1)d_1^2d_2^2 + a(d_1^2 - (d_1^2+d_2^2)\chi \right] 
\\ &\quad \times \left[ \sigma \left(K_c \det{A} (1-\chi)\phi \chi + D_{tr}(\phi) \left[d_1^2 d_2^2 (1-\chi)^2 +a(d_1^2 + d_2^2)(1-\chi)(\chi) + d_1^2d_2^2 \chi^2 \right] \right) \right]^{-1} \\[1em]
\sigma' &= -1 - \rho_s \phi, \\[1em]
\int_0^h \phi \, dz &= \phi_0, \quad
\frac{1}{\phi_0} \int_0^h \phi \chi \, dz = \chi_0, \quad
\sigma(0) = 1 + \rho_s \int_0^z \phi(s) \, ds, \quad
\sigma(h) = 0.
\end{split}
\end{align}

\par By solving the ODE system \eqref{eq:ODE_phi_chi_sigma}, we obtain vertical profiles for $\phi$, $\chi$, and $\sigma$, which coupled with Eq.~\eqref{eq:velocity}, determines the velocity profile in the downstream direction of the mixture. We scale the profile quantities by the film height, so that $\tilde{u} = h^2 u$. By integrating the scaled velocity in the normal direction to the plane, we define the fluxes 
\begin{equation}\label{eq:fluxes}
 f(\phi_0, \chi_0) = \int_0^1 \tilde{u} ds, \quad g_i(\phi_0, \chi_0) = \int_0^1 \tilde{u} \phi_i ds.
\end{equation}
With the change of variables $s = z/h$, we obtain the following hyperbolic system of conservation laws that govern the height profiles of the clear fluid and particle mixtures as it flows down the track:
\begin{equation}\label{eq:PDE}
 \frac{\partial h}{\partial t} + \frac{\partial}{\partial x} (h^3 f(\phi_0, \chi_0)) = 0 \quad
 \frac{\partial (h \phi_0,i)}{\partial t} + \frac{\partial}{\partial x} (h^3 g_i(\phi_0, \chi_0)) = 0 \quad i = 1,2.
\end{equation}

\section{Analysis}\label{sec:analysis}
In Sec.~\ref{subsec:equilibrium} we analyze equilibrium model presented in Eq.~\eqref{eq:system}. We observe larger particles tend to settle at the top as seen in example solutions presented in Fig.~\ref{fig:equilibrium} and we will show in Proposition \ref{prop:brazil_nuts} why we expect larger particles to settle at the top. In Sec.~\ref{subsec:conservation} we discuss the general structure of solutions to systems of hyperbolic conservation laws and the types of shocks we expect to see in the solution.
\subsection{Equilibrium Model}\label{subsec:equilibrium}
\par First, we discuss how the fluid and particles equilibriate in the $z$ direction according to system \eqref{eq:ODE_phi_chi_sigma}. As in \cite{murisic_particle-laden_2011} and \cite{lee_equilibrium_2015}, we define a solution to \eqref{eq:ODE_phi_chi_sigma} as \textit{ridged} if particles of either species accumulate at the top of the profile (if $\phi \to \phi_m$ as $s \to 1$). On the other hand, we define a solution as \textit{settled} if the particle concentration decreases monotonically as $s$ increases. This leads to a layer of clear fluid on top of the profile ($\phi_i = 0$ for all $s \geq s^*$ for some $s^* < 1$).

\begin{figure}[h]
 \centering
 \subfloat[Ridged Regime at $\chi_0 = 0.3$]{\label{subfig:eq_ridged}\includegraphics[width = 0.5\linewidth]{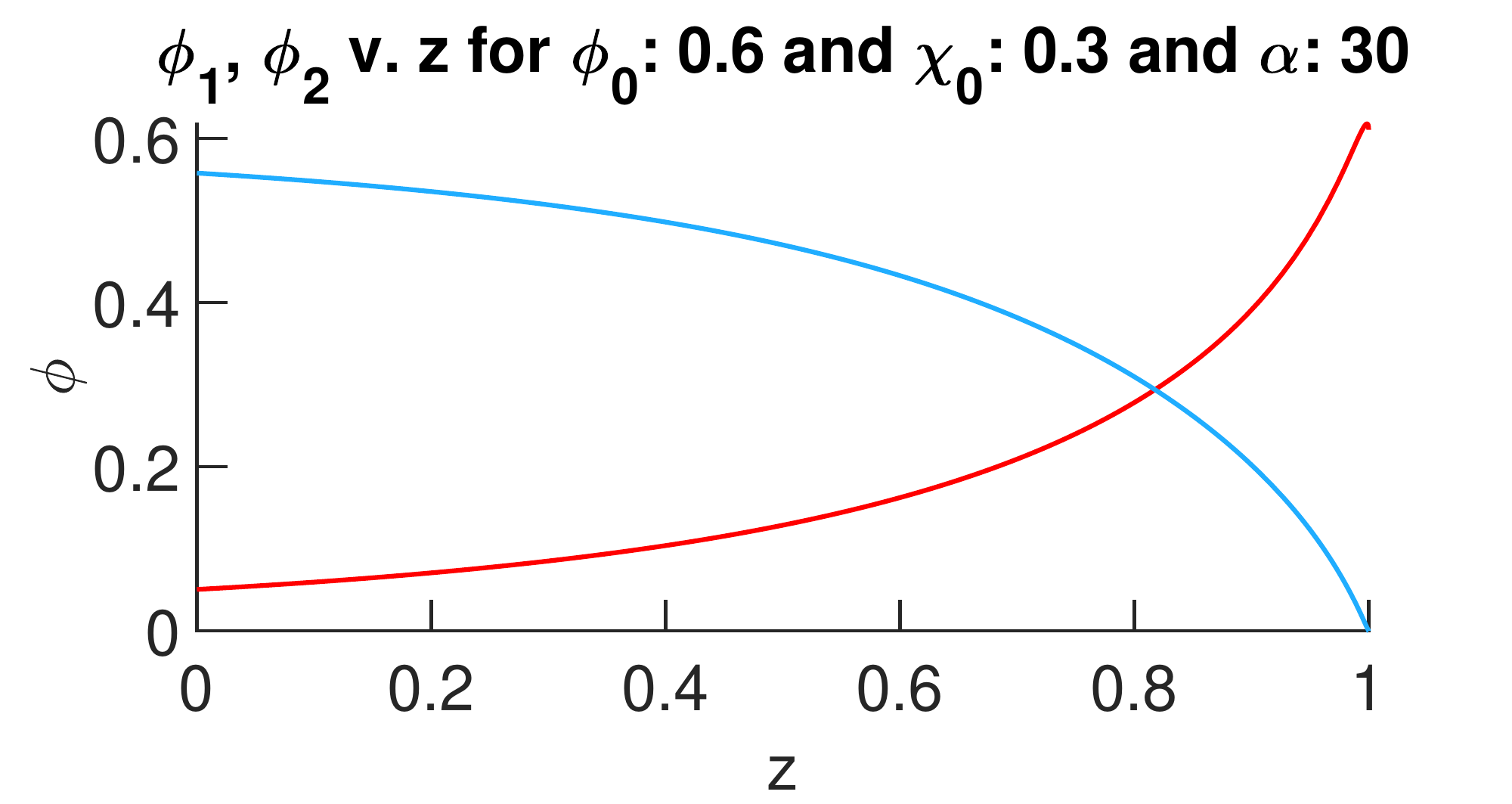}} 
 \subfloat[Ridged Regime at $\chi_0 = 0.7$]{\label{subfig:eq_ridged_small}\includegraphics[width = 0.5\linewidth]{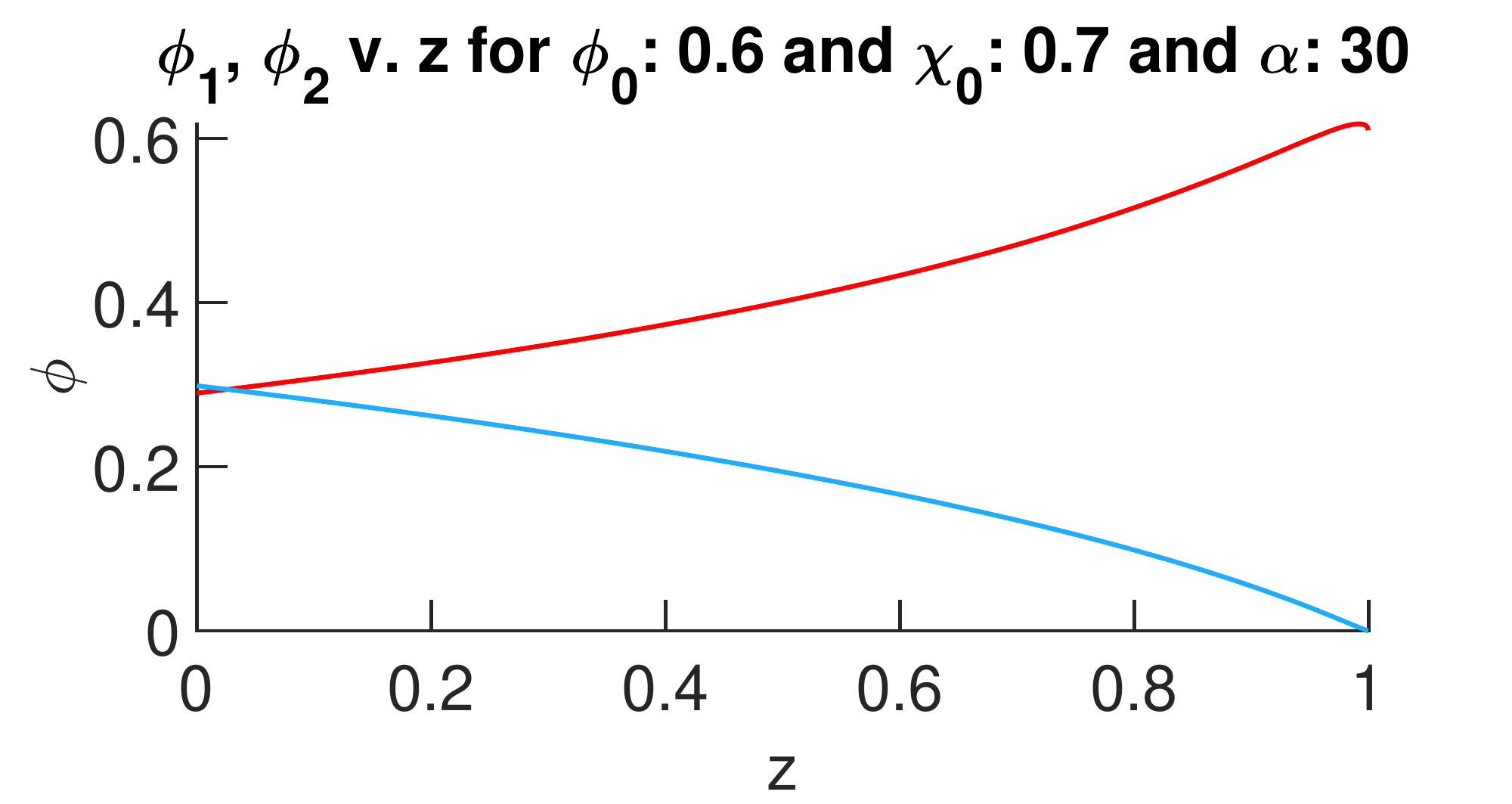}} \\
 \subfloat[Settled Regime at $\chi_0 = 0.3$]{\label{subfig:eq_settled}\includegraphics[width = 0.5\linewidth]{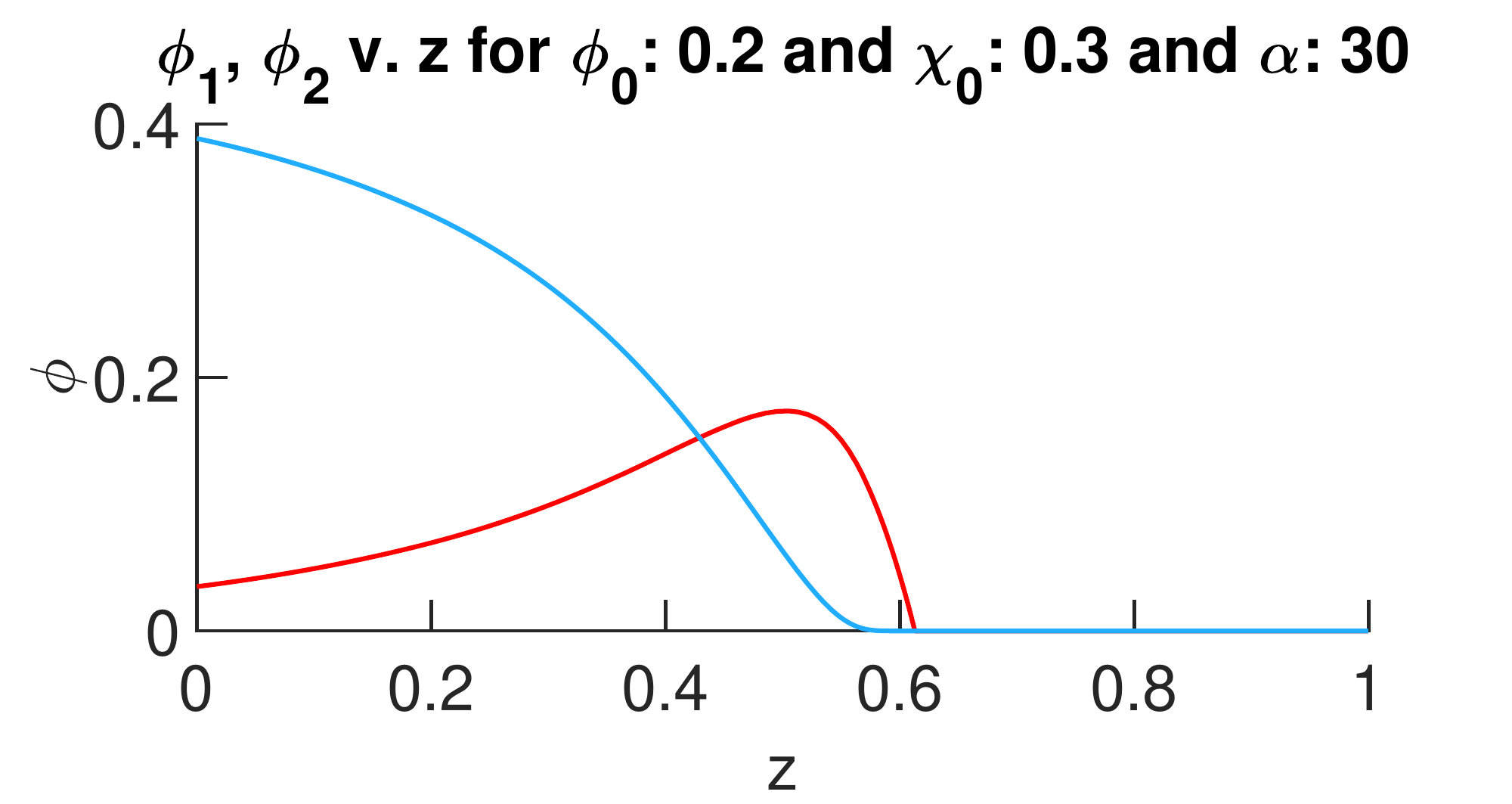}}
 \subfloat[Settled Regime at $\chi_0 = 0.7$]{\label{subfig:eq_settled_small}\includegraphics[width = 0.5\linewidth]{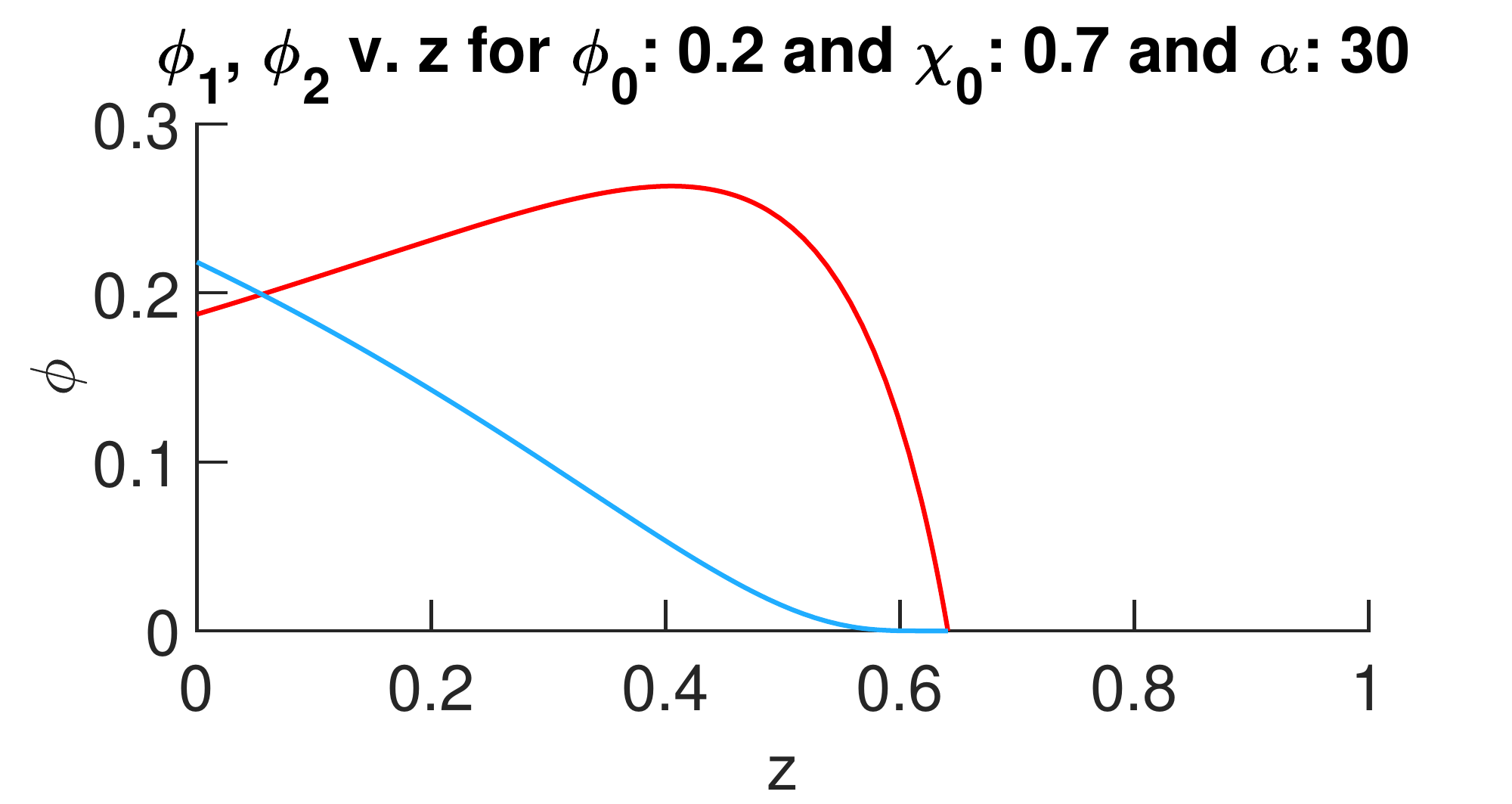}}
 \caption{Example equilibrium profiles that solve system \eqref{eq:ODE_phi_chi_sigma} for the given parameters. Figs.~\ref{subfig:eq_ridged} and \ref{subfig:eq_ridged_small} show the case of a concentrated suspension and Figs.~\ref{subfig:eq_settled} and \ref{subfig:eq_settled_small} show a dilute suspension. The concentration of larger particles is red, and the concentration of smaller particles is blue.}
 \label{fig:equilibrium}
\end{figure}

\par Figure \ref{fig:equilibrium} shows example equilibrium profiles that solve system \eqref{eq:ODE_phi_chi_sigma}. Figure \ref{subfig:eq_ridged} is an example of a ridged profile while Fig.~\ref{subfig:eq_settled} is an example of a settled profile. In Fig.~\ref{subfig:eq_ridged}, as $z$ increases from the substrate to the free surface, the volume fraction of larger particles, $\phi_1$, increases monotonically to the maximum packing fraction, $\phi_m$. On the other hand, the volume fraction of smaller particles, $\phi_2$, decreases monotonically to $0$. Physically, this corresponds to particles on top of the fluid profile, corresponding to $\phi = \phi_1 + \phi_2 \neq 0$ when $z = 1$. Additionally, only larger particles are present at the top of the fluid profile as $\phi_2 = 0$ when $z = 1$. On the other hand, in Fig.~\ref{subfig:eq_settled}, both concentrations of particles, $\phi_1$ and $\phi_2$, vanish as $z$ increases from the substrate to the free surface, eventually reaching $0$ before $z = 1$. Physically, we see a mix of larger and smaller particles in the fluid profile, but the top of the fluid profile contains only clear fluid.

\par As in the single species and bidensity cases, different choices of integral conditions ($\phi_0$ and $\chi_0$) lead to different behaviors for the equilibrium solution. In the bidensity case, the authors of \cite{wong_conservation_2016} are able to show the ODE for $\chi'$ takes the essential form of $\chi(1-\chi)$, and subsequently $\chi$ is an monotonically increasing function. While this result is intuitive in the bidensity case as it implies lighter particles will rise to the top, we show a similar result for the bidisperse case, instead showing larger particles rise to the top.

\begin{proposition}[Brazil-nut Proposition]\label{prop:brazil_nuts}
 Given an integral condition $\phi_0 \in (0, \phi_m)$, $\chi_0 \in (0,1)$ and assuming $0 \leq \phi \leq 1$ and $\det A \geq 0$, then the solution $\chi(s)$ to system \eqref{eq:ODE_phi_chi_sigma} is monotonically increasing.
\end{proposition}
Note the assumption that $0 \leq \phi \leq 1$ is natural as $\phi$ represents the particle concentration. However, the assumption $\det A \geq 0$ requires some explanation. Unfortunately, it is not always true that $\det A \geq 0$. When $d_2/d_1 \lessapprox 0.1$ the determinant is actually negative. In experiments, we choose particle sizes well away from this limit. To prove Proposition \ref{prop:brazil_nuts}, we will use the following lemma that orders the entries of the coupling matrix $A$.
\begin{lemma}\label{lem:A_ij}
 In the coupling matrix $A$ defined in \eqref{eq:A_matrix},
 \[
 \frac{d_2^2}{4} < A_{12} < \frac{d_1^2}{4}
 \]
 where $A_{12}$ is the off diagonal entry.
\end{lemma}
The proof of lemma \ref{lem:A_ij} can be found in the appendix. We now proceed with the proof of proposition \ref{prop:brazil_nuts}.
\begin{proof}
 Our goal is to determine the equilibrium points of $\chi'$ in terms of $\chi$ and the sign on the interval between said equilibrium points. We note the denominator of $\chi'$ is positive (as we assume $\det A$ is positive). This leaves us with examining the numerator of $\chi'$, where we immediately note equilibria at $\chi = 0$ and $\chi = 1$. Since we assume $\phi < 1$, we only need to check the expression $(2\chi -1)d_1^2d_2^2 + a(d_1^2 - (d_1^2+d_2^2)\chi$, which is a linear expression in $\chi$ and can be rewritten as $(2d_1^2d_2^2 - a(d_1^2 + d_2^2)) \chi + (ad_1^2 - d_1^2 d_2^2)$. Recall that $a$ is the off diagonal entry of the coupling matrix $A$ divided by $\frac{1}{4}$ as stated in equation \eqref{eq:A_offdiag}. We compute this linear term at values $\chi = 0$ and $\chi = 1$. At $\chi = 0$, the expression evaluates to $ad_1^2 - d_1^2d_2^2 = d_1^2 (a - d_2^2)$. As established in Lemma \ref{lem:A_ij}, $d_2^2 < a$, so the expression is positive at $\chi = 0$. Similarly, at $\chi = 1$, the expression evaluates to $d_1^2 d_2^2 - a d_2^2 = d_2^2 (d_1^2 - a)$. But, since $d_1^2 > a$, the expression is also positive at $\chi = 1$. Since the linear expression is positive at $\chi = 0$ and $\chi = 1$, it must be positive for all $\chi \in (0,1)$. Thus, $\chi'$ is positive for $\chi \in (0,1)$, and $\chi$ as a solution to system \eqref{eq:system} is monotone increasing. 
\end{proof}
\par Proposition \ref{prop:brazil_nuts} establishes that regardless of whether the solution is ridged or settled, the ratio of larger particles to the total number of particles increases as the height of the fluid profile increases. We can heuristically explain this observation by arguing sedimentation flux, which would otherwise promote stokes settling, is counterbalanced by the shear and tracer fluxes when on an incline. As the inclination angle decreases, the shear rate decreases and so do the shear and tracer fluxes. The settling flux does not scale with shear rate and thereby dominates the other two. In practice, the smaller particles vanish at the top of the equilibrium profiles which leaves only the larger particles at the top of the particle profile as in Fig.~\ref{fig:equilibrium}.

\par Our model is consistent with the observations in granular convection and inclined granular flows that larger particles rise to the top, and provides a mathematical explanation in the equilibrium theory for why this should be the case. We note that our model does not include mechanisms that try to capture the ability of smaller particles being able to fit into the void space larger particles are not, unlike in other works \cite{rosato_why_1987, savage_particle_1988, thornton_three-phase_2006}. In particular, our sedimentation term in Eq.~\eqref{eq:J_settling} for a particle species $i$ does not explicitly depend on the concentration of the other particle species; we only include a hindered settling term that depends on the presence of both species of particles. While the maximum packing fraction $\phi_m$ is a function of $\chi$, it does not impact differently sized particles differently. Our model also does not include wall effects which \cite{knight_vibration-induced_1993} and \cite{knight_experimental_1996} propose as one of the potential explanations behind the Brazil-nut effect.

\par While we do not claim to fully explain the Brazil-nut effect, we find it encouraging the Brazil-nut effect naturally arises in our model without including terms that try to account for it. However, our model assumes the particles are suspended in a viscous fluid and many simplifying assumptions were made with this assumption. It is unclear how dependent our Brazil-nut conclusion is on such assumptions. Furthermore, the Brazil-nut effect is obviously observed in more general settings as well. Still, it is promising the Brazil-nut effect arises in our model without specific consideration.

\subsection{Conservation Law Model}\label{subsec:conservation}
\par With the ability to compute the equilibrium profiles that solve system \eqref{eq:ODE_phi_chi_sigma}, we are also able to compute the fluxes in Eq.~\eqref{eq:fluxes} and thus advance the PDE system described in Eq.~\eqref{eq:PDE}. We establish the structure of solutions to a three by three system of hyperbolic conservation laws as presented in \cite{lax_1_1973}. We write the system in vector form so let $U = (h, h\phi_{0,1}, h\phi_{0,2})^T$ and $F(U) = h^3(f,g_1,g_2)^T$ resulting in the system
\begin{equation}\label{eq:system}
 \frac{\partial U}{\partial t} + \frac{\partial}{\partial x}(F(U)) = 0.
\end{equation}
To understand the solution structure, we study the solution of Eq.~\eqref{eq:system} under Riemann initial conditions:
\begin{equation}\label{eq:Riemann}
 U(x,0) = \begin{cases}
 U^L & x < 0 \\
 U^R & x > 0
 \end{cases}
\end{equation}
with the superscripts denoting the left or right end states respectively. We call the system \eqref{eq:system} with initial conditions \eqref{eq:Riemann} the Riemann problem. Although our experiments and thus our numerical simulations have a finite volume initial condition, understanding of the solution structure for Riemann initial condition can be bridged to understanding for finite volume initial condition as done in the monodisperse case in \cite{wang_rarefaction-singular_2015}. Denoting $J$ as the Jacobian of the system, we order the eigenvalues $0 < \lambda_1 < \lambda_2 < \lambda_3$ assuming the eigenvalues are real and positive. The eigenvalues depend on the state $U$.

\par The weak solution to the Riemann problem transports the discontinuity from $U^L$ to $U^R$ with a speed $s$ prescribed by the Rankine-Hugoniot jump condition:
\begin{equation}\label{eq:RH}
F\left(U^L\right) - F\left(U^R\right) = s\left(U^L - U^R\right)
\end{equation}
We impose the Lax entropy condition to define what are considered permissible shocks:
\begin{equation}\label{eq:Lax}
\lambda_k\left(U^L\right) > s > \lambda_k\left(U^R\right), \quad \lambda_{k+1}\left(U^R\right) > s > \lambda_{k-1}\left(U^L\right)
\end{equation}
for some $k$. So, for a given state $U$, the Hugoniot locus is the smooth curve that passes through $U$ and consists of all other states that solve equation \eqref{eq:RH} under the Lax entropy condition \eqref{eq:Lax}. The state $U$ can be taken as either the left or right state.

\par For a given left state $U^L$ and right state $U^R$ we don't expect a single $s$ to satisfy the Rankine-Hugonoit jump condition in \eqref{eq:RH} as that would require each component of the equality to be satisfied. This corresponds geometrically to there being no direct Hugoniot curve connecting the states $U^L$ and $U^R$. However, we can find intermediary states such that there is a sequence of intermediary states such that there exists a connection between $U^L$ and $U^R$ by traveling along the Hugoniot curves of $U^L$, the intermediary states, and $U^R$. In general, we expect there to be three intermediary states which leads to the triple shock structure we see in the settled regime \cite{wong_conservation_2016}. However, in the ridged regime, we see a singular shock, a shock whose height is singular, instead. In the monodisperse \cite{wang2014shock} and bidensity case \cite{wong_conservation_2016}, this is due to the Hugoniot curves of the states approaching each other asymptotically in the far field instead of intersecting. We pose as an open problem that as the system transitions from the settled to the ridged regime, the Hugoniot curves of the states approach each other asymptotically in a way such that a single singular shock forms, as opposed to the multiple shock structure in the bidensity case. In the next sections, we present the laboratory experimental design and a few experiments where the Brazil-nut effect and evolution of the particle and liquid fronts can be seen. 

\begin{figure}[h]
 \centering
 \includegraphics[width=0.5\linewidth]{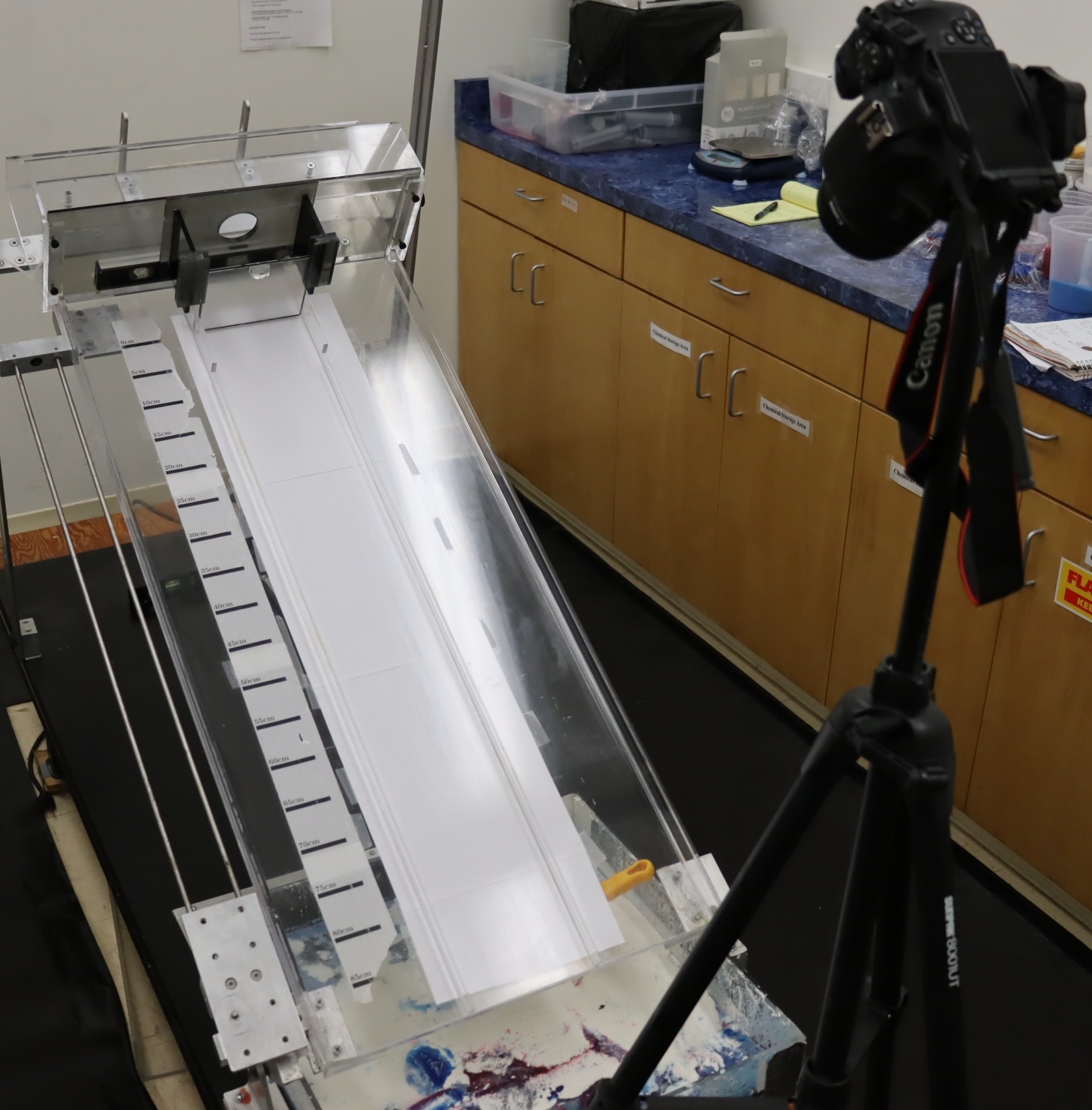}
 \caption{Setup for the constant volume inclined plane experiments. }
 \label{fig:setup}
\end{figure}

\section{Experiments}\label{sec:experiment}

The goal of the experiments is to capture the separation of different-sized particles suspended in a liquid thin-film moving down an incline. Fig.~\ref{fig:setup} shows the inclined track and table setup for the constant volume inclined plane experiments. The experiments are carried out on a flat one-meter-long acrylic track that can be adjusted at angles of $10-65^\circ$. This is the same track used in previous work \cite{zhou2005theory, murisic_particle-laden_2011, wong_conservation_2016}. Before running the experiment, the track is wiped with paper towels, and a gate is placed at the front of the reservoir seen at the top of the track. A well-mixed slurry is poured behind the gate and then quickly released to start the experiment. An aerial view of the experiment is filmed to capture the distance of each front (liquid, large particles, small particles) as seen in the sample frame in Fig.~\ref{subfig:settled}. 

The particles are hard, spherical beads made with soda lime glass (Ceroglass). They are $65-95 \%$ round with a density of $2.45-2.50 \mathrm{~g} / \mathrm{cm}^3$. The beads used were `deco beads' which are colored and processed by the manufacturer. The two species sizes we choose were light blue with a diameter of $0.2-0.4 \mathrm{~mm}$ and red with a diameter of $0.5-0.75 \mathrm{~mm}$. The particles are mixed with PDMS silicone oil (Clearco Products) with a medium viscosity of $1,000 \mathrm{cSt}$ and a density of $0.971 \mathrm{~g} / \mathrm{cm}^3$. In the next section, we show a few experiments and compare them with the model.

\section{Model case studies compared with laboratory experiments}\label{sec:discussion}

\par We make qualitative comparisons between numerical simulations and physical experimental data. Figures \ref{subfig:settled_25_exp}, \ref{subfig:settled_50_exp}, and \ref{subfig:settled_75_exp} are experiments and Figs.~\ref{subfig:settled_25_num}, \ref{subfig:settled_50_num}, and \ref{subfig:settled_75_num} are simulations done in the settled regime with parameters $\alpha = 20^{\circ}$ and $\phi_0 = 0.4$. The parameter $\chi_0$ is varied across the different runs. Figure \ref{subfig:ridged_exp} shows an example of the ridged regime with parameters $\alpha = 55^{\circ}$, $\phi_0 = 0.4$, and $\chi_0 = 0.5$.
For the numerical simulations the large particle diameter used is $0.625 \mathrm{~mm}$ and the small particle diameter used is $0.200 \mathrm{~mm}$. 

\par For the numerical simulation, we assume a fixed volume initial condition and use an upwind scheme. The fixed volume is a rectangle of volume $\phi_0$ supported from $-8.4 \cdot 10^{-2} \mathrm{m}$ to $0 \mathrm{m}$ to model how the slurry initially sits in the reservoir. The numerical domain ranges from $-8.4 \cdot 10^{-2}$ to $1.2$ with $514$ grid points. Similar to previous numerical simulations \cite{murisic_particle-laden_2011} and \cite{wong_conservation_2016}, we assume the slurry remains well-mixed until a transition time $t^* = 150$s and which afterwards the particles are assumed to evolve according to equation \eqref{eq:PDE}.

\par Overall, we observe excellent qualitative agreement between the numerical simulations and experimental data, with our model accurately capturing the key qualitative features of the experiments. First, the larger particles in both the settled and ridged regimes are always at the front of the particle ridge. This is explained with the understanding of the Brazil-nut effect established in Proposition \ref{prop:brazil_nuts}. Since larger particles rise to the top according to the established equilibrium theory, the higher velocity field at the free surface of the fluid moves the larger particles faster downstream to the front of the slurry. When examining the side profiles of our experiments, the larger particles congregate at the top as shown in Fig.~\ref{fig:sideprofile}.

\par When looking at the experiments in Fig.~\ref{fig:tile_comparison}, recall the larger (red) particles will appear on top of the smaller (blue) particles. So in aerial profiles, it can be harder to view the separation when there's a high concentration of larger particles. Also, the simulations use a fixed diameter for each particle size. In the experiments, we have distributions of particle sizes and we use the mean to represent the diameter for comparison. Having distributions of sizes in the experiment can account for the lack of distinct separation.

\subsection{Settled regime, $\chi_0 = 0.25$}
\par We begin with comparisons between the experimental data and numerical simulations in the settled regime. In Figs.~\ref{subfig:settled_25_num} and \ref{subfig:settled_25_exp}, we have the case with the parameters above and $\chi_0 = 0.25$. In the numerical simulation in Fig.~\ref{subfig:settled_25_num}, the simulation predicts a settled regime. Behind the clear fluid front, we expect a sharp yet thin front of the larger red particles, followed by smaller light blue particles that make up the bulk of the slurry upstream. This is exactly what we see in Fig.~\ref{subfig:settled_25_exp}, although the clear fluid front predicted in the numerical simulation appears much less than what is actually observed in experiments. However, the strong fingering instability seen in Fig.~\ref{subfig:settled_25_exp} may contribute to this discrepancy.

\subsection{Settled regime, $\chi_0 = 0.5$}
\par Next, in Figs.~\ref{subfig:settled_50_num} and \ref{subfig:settled_50_exp}, we have the case where $\chi_0 = 0.5$. The numerical simulation in Fig.~\ref{subfig:settled_50_num} is similar to the previous case where we have a settled regime and a front of the larger red particles leading the particle front. However, this red front is thicker than the particle front we saw in the case where $\chi_0 = 0.25$. Additionally, the boundary of the larger red particle front smoothly blends with the smaller light blue particles as we continue to go upstream. Again, this is a different from the case in figure \ref{subfig:settled_25_exp} where the difference between the smaller and larger particle fronts is stark. We see in Fig.~\ref{subfig:settled_50_exp} this is the case. Upstream from the clear fluid front is the red particle front. By comparing this particle front in Fig.~\ref{subfig:settled_50_exp} to the particle front in Fig.~\ref{subfig:settled_25_exp}, we do see the particle front in Fig.~\ref{subfig:settled_50_exp} is thicker, and this is corroborated when comparing the numerical results in Fig.~\ref{subfig:settled_50_num} and Fig.~\ref{subfig:settled_25_num}. This large, red particle ridge blends with the small, light blue particles as we go upstream, which again is predicted in the numerical simulations.

\subsection{Settled regime, $\chi_0 = 0.75$}
\par Finally, Figs.~\ref{subfig:settled_50_num} and \ref{subfig:settled_50_exp} show the case where $\chi_0 = 0.75$. Once again, the larger red particles lead the particle front in the numerical simulation shown in Fig.~\ref{subfig:settled_75_num}. Similar to the case where $\chi_0 = 0.5$, the large, red particle front blends with the small, light blue particles as we go upstream and appears well-mixed near the top of the incline. This behavior is shown in Fig.~\ref{subfig:settled_75_exp}. The red particle ridge is distinctly thicker than the particle ridge in Figs.~\ref{subfig:settled_25_exp} and \ref{subfig:settled_50_exp}, and smoothly transitions to a well-mixed blend of red and light blue particles as predicted in the numerical simulations. Overall, our numerical simulations are not only able to predict the flow being settled, but also predict detailed behavior of the particle fronts.

\subsection{Ridged regime}
\par We also have an example of an experiment in the ridged regime in Fig.~\ref{subfig:ridged_exp}. There are not as many distinct fronts to capture in the ridged case compared to the settled case which is why the majority of our comparisons so far have focused on the settled case. We see our numerical simulation in Fig.~\ref{subfig:ridged_num} predicts the flow is in the ridged regime as the particles are present at the very front of the slurry. Again, the particle front is comprised mostly of the larger, red particles. Our numerical simulations also predict a singular shock. In the experimental results shown in Fig.~\ref{subfig:ridged_exp}, we see the particles at the very front of the flow establishing this experiment is in the ridged regime. Furthermore, the particle front is comprised mostly of the larger, red particles as we expect from the numerical simulations. 
\par We do not see the explicit presence of a singular shock at the particle front, which is expected. Our model is a hyperbolic conservation law where singular shocks can occur, while in reality there are surface tension effects which dampen shock formation. We can see in both the simulations and in the experiment the particle concentration at the front is much higher than that in the settled regime, a hallmark of the ridged regime. In the numerical simulations, we are able to determine the particle volume fraction approaches the particle maximum packing fraction. Although we are not able to determine this in the experiment, visual inspection of the red particle ridge in Fig.~\ref{subfig:ridged_exp} suggests the particle volume fraction also approaches the particle maximum packing fraction.

\begin{figure}
 \centering
 \subfloat[Numerical Simulation]{\label{subfig:settled_25_num}\includegraphics[width = 0.48\linewidth]{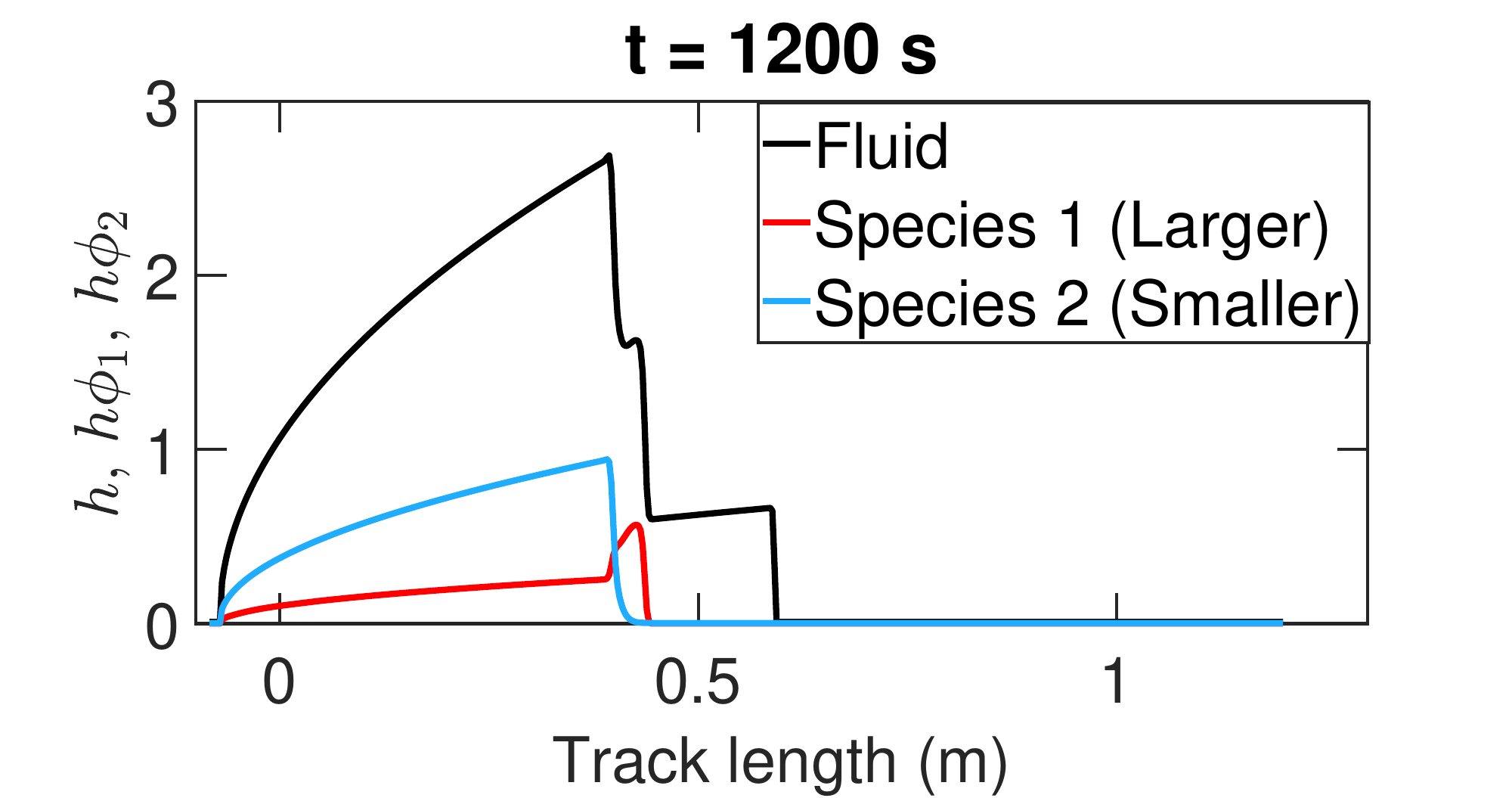}} ~
 \subfloat[Numerical Simulation]{\label{subfig:settled_50_num}\includegraphics[width = 0.48\linewidth]{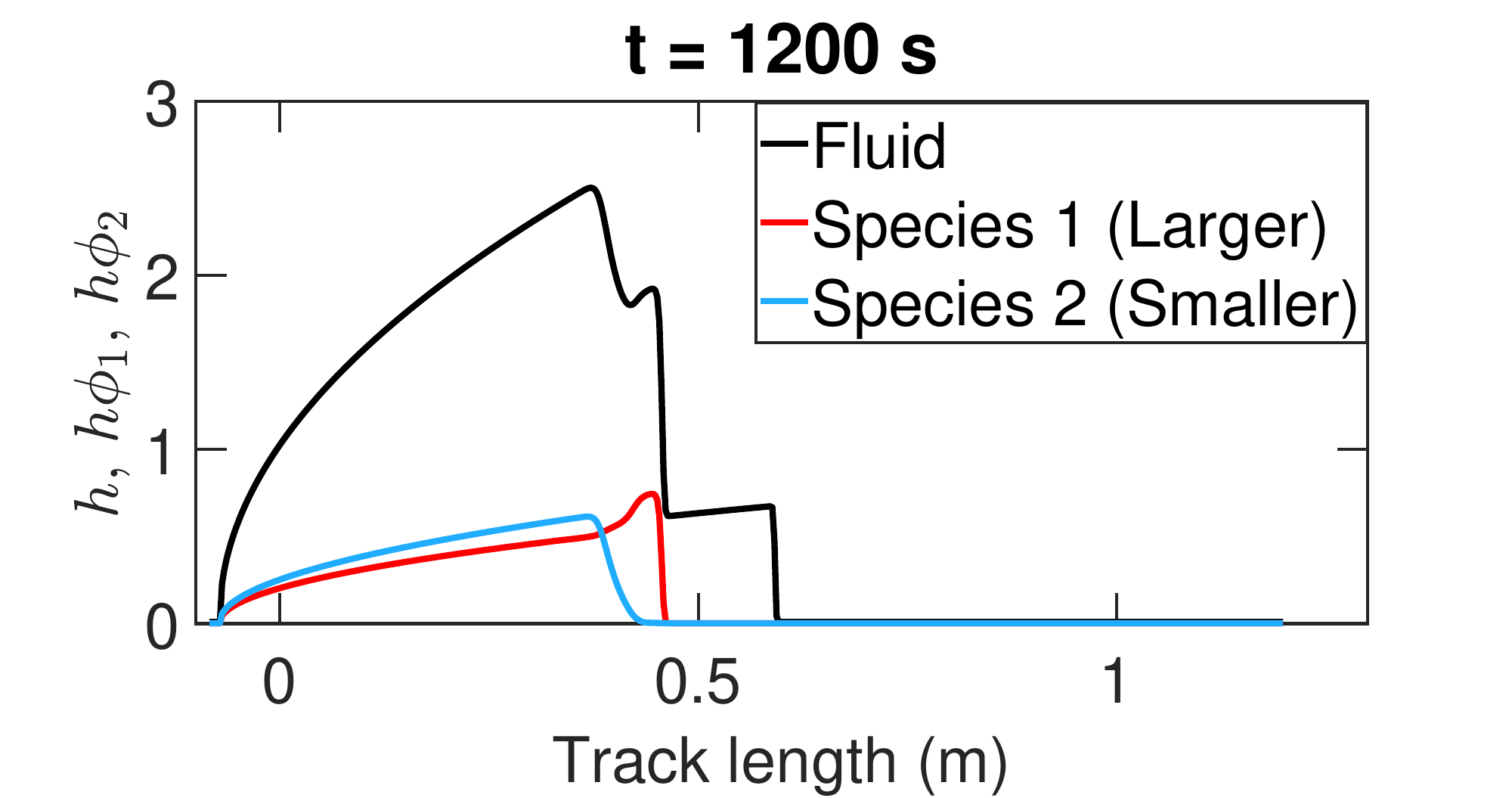}} \\
 \subfloat[Experiment]{\label{subfig:settled_25_exp}\includegraphics[width = 0.43\linewidth]{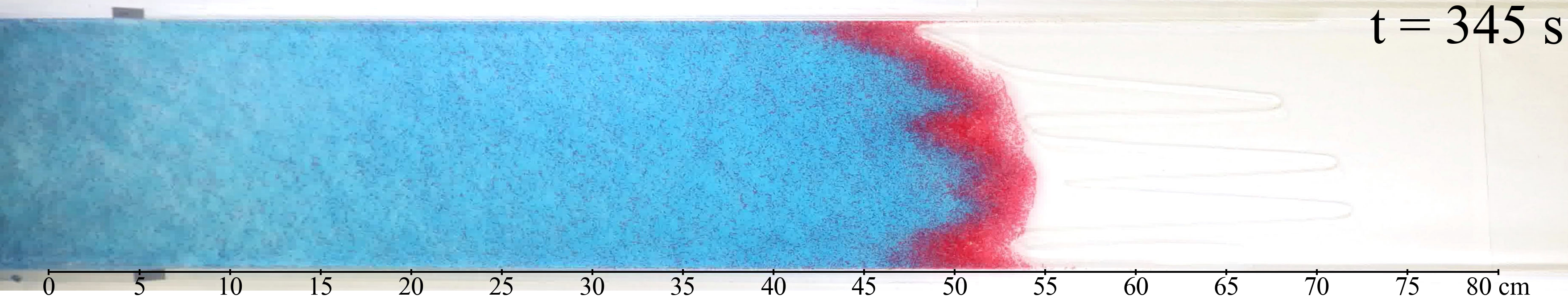}} \hspace{4em}
 \subfloat[Experiment]{\label{subfig:settled_50_exp}\includegraphics[width = 0.43\linewidth]{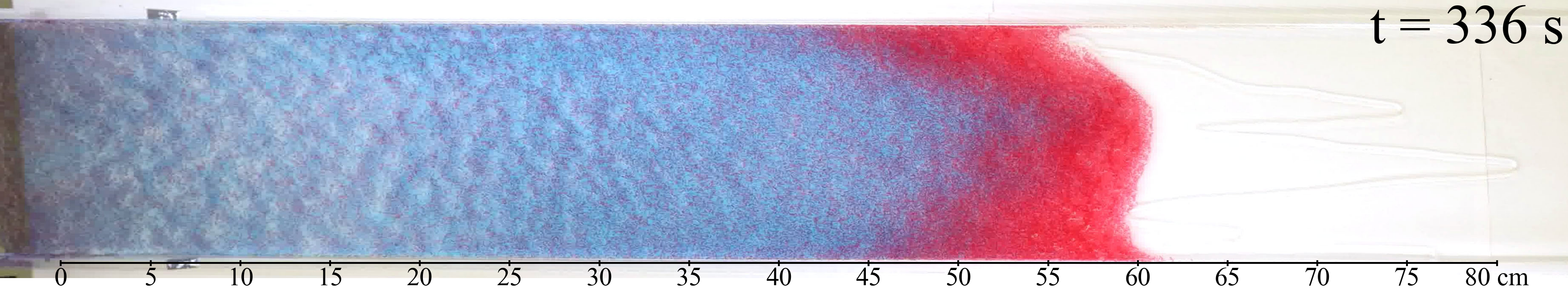}} \\
 \subfloat[Numerical Simulation]{\label{subfig:settled_75_num}\includegraphics[width = 0.48\linewidth]{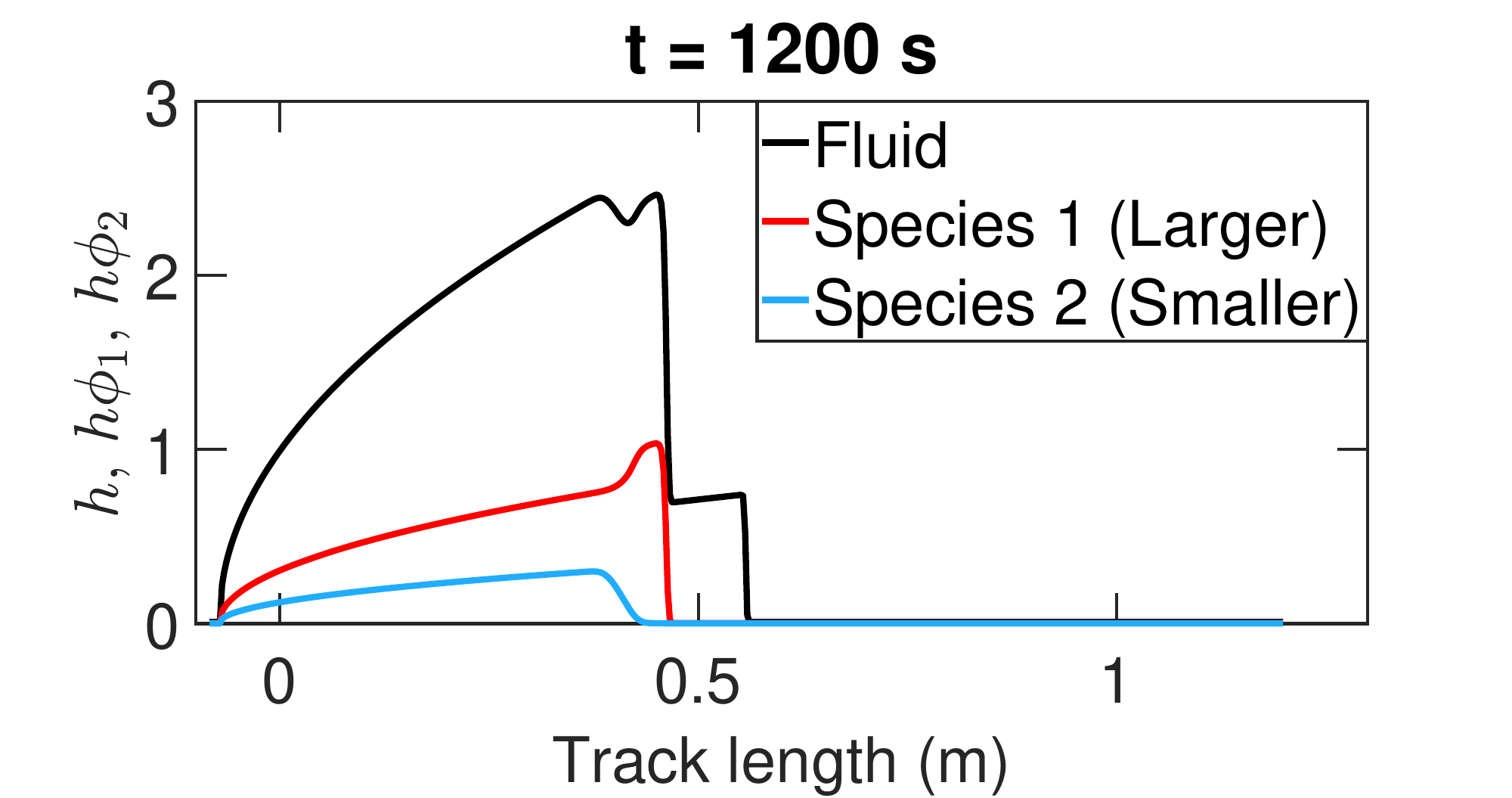}} ~
 \subfloat[Numerical Simulation]{\label{subfig:ridged_num}\includegraphics[width = 0.48\linewidth]{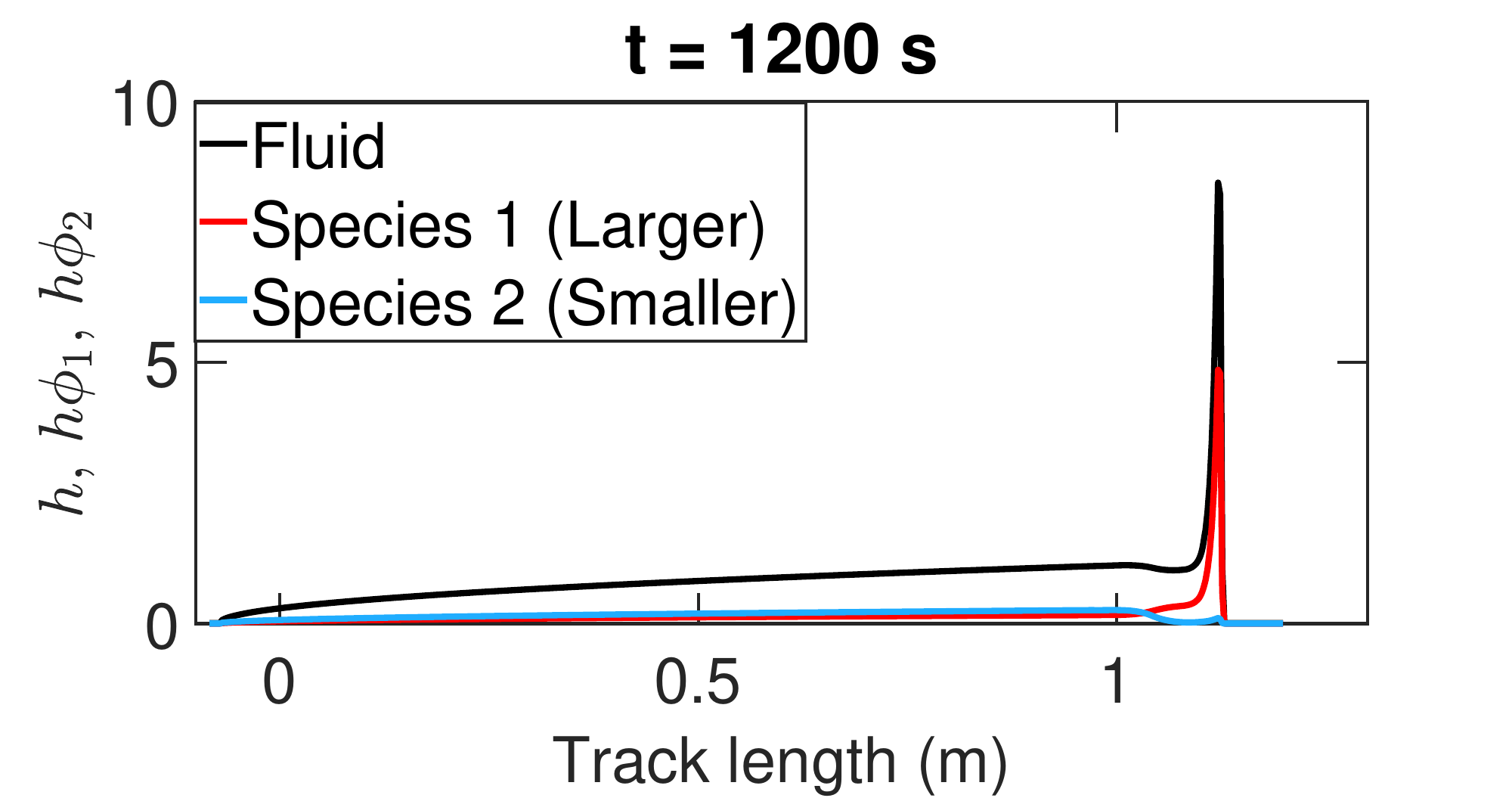}} \\
 \subfloat[Experiment]{\label{subfig:settled_75_exp}\includegraphics[width = 0.43\linewidth]{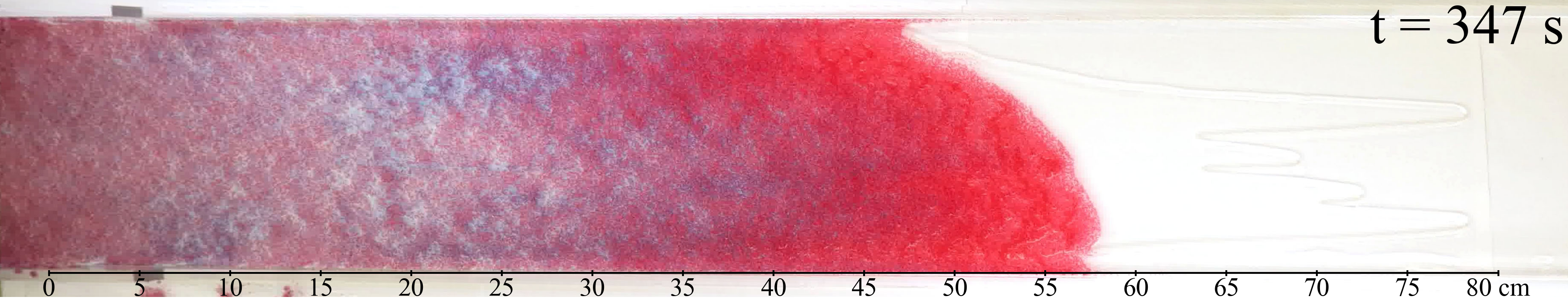}} \hspace{4em}
 \subfloat[Experiment]{\label{subfig:ridged_exp}\includegraphics[width = 0.43\linewidth]{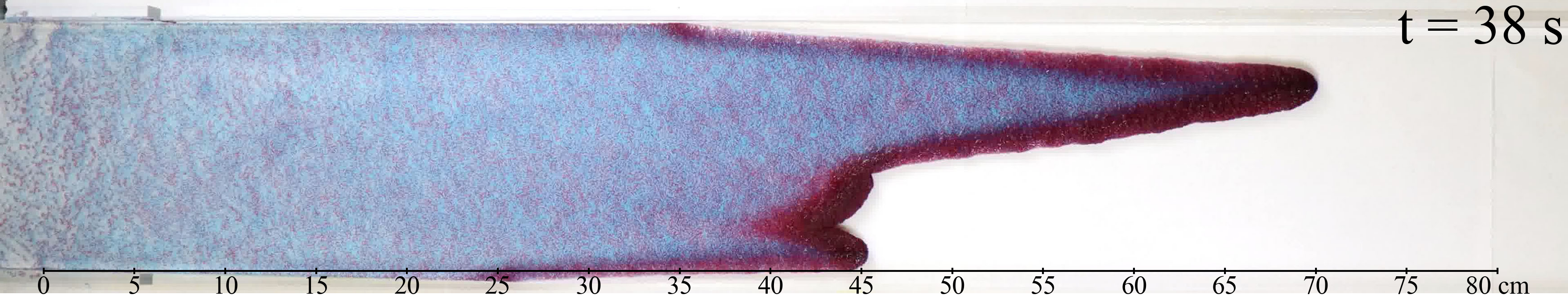}}
 \caption{Comparison between numerical simulations and experimental results for various parameter choices. Figs.~\ref{subfig:settled_25_num} and \ref{subfig:settled_25_exp}: $\phi_0 = 0.4$, $\chi_0 = 0.25$. Figs.~\ref{subfig:settled_50_num} and \ref{subfig:settled_50_exp}: $\phi_0 = 0.4$, $\chi_0 = 0.5$. Figs.~\ref{subfig:settled_75_num} and \ref{subfig:settled_75_exp}: $\phi_0 = 0.4$, $\chi_0 = 0.75$. Figs.~\ref{subfig:ridged_num} and \ref{subfig:ridged_exp}: $\phi_0 = 0.55$, $\chi_0 = 0.5$. The track is $1$ meter long. Note the simulation times do not match exactly to the experimental time.}
 \label{fig:tile_comparison}
\end{figure}

\section{Conclusion}\label{sec:conclusion}
\par In this work, we have a presented a conservation law model for particle laden flows on an incline of two differently sized species of particles by extending the thin film equilibrium and conservation law model explored by \cite{murisic_particle-laden_2011} and \cite{wong_conservation_2016}. For our equilibrium model, we expand upon the flux balances presented by \cite{wong_conservation_2016} but include coupling effects caused by the different sizes of particles introduced by \cite{kanehl_hydrodynamic_2015} and the dependence of the maximum packing fraction based on the presence of differently sized particles \cite{shauly_shear-induced_1998}. In our analysis of the equilibrium model, we find larger particles are expected to rise to the top of the profile regardless of whether the profile is in the settled or ridged regime. This observation is supported by our experimental results where both the larger particles rise to the top of side profiles during experiments and where larger particles congregate at the end of the particle front. This finding is similar to the Brazil-nut effect observed in granular flows. In our model, the Brazil-nut effect arises as a natural consequence in the equilibrium theory without having to make specific concessions for it when modeling the flux terms. Our setting is quite different than traditional granular flow as the particles are submerged in a viscous fluid and we are able to reduce the complexity of the model by leveraging thin film approximations. However, we pose as an open question if the Brazil-nut effect can arise naturally in granular flow models without the express need to include it in the mathematical model.

\par We next perform numerical simulations of the conservation law with a fixed volume initial condition and compare them with experimental results. Overall, our qualitative agreement is excellent. Not only is our model able to accurately predict whether the flow is in the settled or ridged regime, but our model is also able to predict finer features of the flow particularly in the settled regime such as particle ridge width and well-mixed behavior upstream. This qualitative matching is encouraging and suggests the overall modeling approach is sound.

\par However, we believe there are improvements to be made in the model. In our equilibrium models, although we adapt the shear-induced migration for two particle species of different sizes from \cite{kanehl_hydrodynamic_2015}, we use the same settling and tracer fluxes from \cite{wong_conservation_2016} except we change the particle diameters to be different depending on the species. It is unclear whether more accurate modeling needs to be done for these two effects. For example, the original formulation of tracer flux in \cite{wong_conservation_2016} has the property that the sum of tracer fluxes of the two species cancels each other out, which is not true in our formulation of the tracer flux because of the differing diameters of the particle species. Additionally, as stated previously, our settling flux does not include size effects. It would be interesting to include size effects in the settling flux and see whether the same qualitative features emerge in the numerical simulations.

\par This goal of this work is to provide a fundamental model that accurately describes particle-laden flow behavior in thin-films. Although we do not qualitatively match the front positions between numerical simulations and experiments, we are still able to present a model that is simple enough to solve numerically and make predictions that well match the behavior of the experiments. We also rigorously explain the phenomena of large particles rising to the free surface and appearing at the front of the particle region by proving that larger particles proportionally increase as the distance from the substrate increases. We are interested to see how the modeling choices or improvements suggested in the previous paragraph lead to different qualitative observations or quantitative agreements.

\begin{acknowledgments}
This material is based upon work supported by the U.S. National Science Foundation under award No. DMS-2407006. This
work is also supported by Simons Math + X Investigator Award number 510776. Sarah C. Burnett was supported by the 2022
L’Oréal USA for Women in Science Postdoctoral Fellowship. Preliminary experimental work supporting this research was done during the Research Experience for Undergraduates (REU) program at University of California, Los Angeles during Summer 2022,  2023, \& 2024. 
\par 
\end{acknowledgments}

\appendix

\section{}
We give the proof of lemma \ref{lem:A_ij}, stated again below.
\begin{lemma}
 In the coupling matrix $A$ when $d = 3$ defined in \eqref{eq:A_matrix},
 \[
 \frac{d_2^2}{4} < a_{12} < \frac{d_1^2}{4}
 \]
 where $a_{12}$ is the off diagonal entry.
\end{lemma}
\begin{proof}
 We note we can divide $d_1^2/4$ from all terms in the desired inequality. For the off diagonal entry, we can factor out $d_1^2/4$ yielding
 \begin{align*}
 a_{12} = \frac{d_1^2}{4} \frac{(1 + d_2/d_1)^2}{2^{d+1}}\frac{(1 + d_2/d_1)^d}{1 + (\frac{d_2}{d_1})^d}.
 \end{align*}
 So, we wish to show the equivalent inequality
 \[
 \frac{d_2^2}{d_1^2} < \frac{(1 + d_2/d_1)^2}{2^{d+1}}\frac{(1 + d_2/d_1)^d}{1 + (\frac{d_2}{d_1})^d} < 1
 \]
 Now that we have reduced the inequality to one functional variable $\frac{d_2}{d_1}$, we set $\frac{d_2}{d_1} = 1 - \epsilon$, with $0 < \epsilon < 1$. This is justified as we assume $0 < d_2 < d_1$. Thus, we simplify the modified off diagonal entry as 
 \[
 \tilde{a}_{12} = \frac{(2 - \epsilon)^2}{2^{d+1}} \frac{(2 - \epsilon)^d}{1 + (1 - \epsilon)^d}.
 \]
 To see that $\tilde{a}_{12} < 1$, we observe
 \begin{align*}
 \frac{(2 - \epsilon)^2}{2^{d+1}} \frac{(2 - \epsilon)^d}{1 + (1 - \epsilon)^d} &< \frac{2^{d + 2}}{ 2 \cdot 2^{d+1} } \\
 &< 1.
 \end{align*}
 The other inequality, to show that $\tilde{a}_{12} > (1 - \epsilon)^2$ is trickier. We will show that when $\epsilon \in (0,1)$, the inequality
 \[
 \frac{(2 - \epsilon)^2}{2^{d+1}} \frac{(2 - \epsilon)^d}{1 + (1 - \epsilon)^d} - (1 - \epsilon)^2 > 0
 \]
 holds. This is equivalent to showing the polynomial
 \[
 (1 + x^5) - 2^4 x^2 (1 + x^3)
 \]
 where $x = 1 - \epsilon$ is positive on the domain $(0,1)$. We may factor the polynomial as 
 \[
 -(x-1)(x+1)(15x^3 - 5x^2 + 5x + 1).
 \]
 We note $-(x-1)(x+1) > 0$ when $x \in (0,1)$, so now all we have to do is verify the sign of $(15x^3 - 5x^2 + 5x + 1)$ on $(0,1)$. To do so, we will show the polynomial is monotone by examining its derivative, $5(9x^2 - 2x + 1)$. Since the discriminant $4 - 4 \cdot 9 \cdot 1 = -32$ is negative, the derivative has no real roots and since it is positive at $x = 0$, the derivative is positive for all $x$. Thus, the original polynomial s monotone increasing. The polynomial evaluated at $x = 0$ is $1$, and thus the polynomial is positive on the interval $(0,1)$, and we conclude $(1 - \epsilon)^2 < \tilde{a}_{12}^2$. This establishes the desired inequality $(1 - \epsilon)^2 < \tilde{a}_{12} < 1$.
\end{proof}

\bibliography{sources}

\providecommand{\noopsort}[1]{}\providecommand{\singleletter}[1]{#1}%
\begin{thebibliography}{42}%
\makeatletter
\providecommand \@ifxundefined [1]{%
 \@ifx{#1\undefined}
}%
\providecommand \@ifnum [1]{%
 \ifnum #1\expandafter \@firstoftwo
 \else \expandafter \@secondoftwo
 \fi
}%
\providecommand \@ifx [1]{%
 \ifx #1\expandafter \@firstoftwo
 \else \expandafter \@secondoftwo
 \fi
}%
\providecommand \natexlab [1]{#1}%
\providecommand \enquote  [1]{``#1''}%
\providecommand \bibnamefont  [1]{#1}%
\providecommand \bibfnamefont [1]{#1}%
\providecommand \citenamefont [1]{#1}%
\providecommand \href@noop [0]{\@secondoftwo}%
\providecommand \href [0]{\begingroup \@sanitize@url \@href}%
\providecommand \@href[1]{\@@startlink{#1}\@@href}%
\providecommand \@@href[1]{\endgroup#1\@@endlink}%
\providecommand \@sanitize@url [0]{\catcode `\\12\catcode `\$12\catcode `\&12\catcode `\#12\catcode `\^12\catcode `\_12\catcode `\%12\relax}%
\providecommand \@@startlink[1]{}%
\providecommand \@@endlink[0]{}%
\providecommand \url  [0]{\begingroup\@sanitize@url \@url }%
\providecommand \@url [1]{\endgroup\@href {#1}{\urlprefix }}%
\providecommand \urlprefix  [0]{URL }%
\providecommand \Eprint [0]{\href }%
\providecommand \doibase [0]{https://doi.org/}%
\providecommand \selectlanguage [0]{\@gobble}%
\providecommand \bibinfo  [0]{\@secondoftwo}%
\providecommand \bibfield  [0]{\@secondoftwo}%
\providecommand \translation [1]{[#1]}%
\providecommand \BibitemOpen [0]{}%
\providecommand \bibitemStop [0]{}%
\providecommand \bibitemNoStop [0]{.\EOS\space}%
\providecommand \EOS [0]{\spacefactor3000\relax}%
\providecommand \BibitemShut  [1]{\csname bibitem#1\endcsname}%
\let\auto@bib@innerbib\@empty
\bibitem [{\citenamefont {Santangelo}\ \emph {et~al.}(2021)\citenamefont {Santangelo}, \citenamefont {Forte}, \citenamefont {De~Falco}, \citenamefont {Chirico},\ and\ \citenamefont {Santo}}]{santangelo2021new}%
  \BibitemOpen
  \bibfield  {author} {\bibinfo {author} {\bibfnamefont {N.}~\bibnamefont {Santangelo}}, \bibinfo {author} {\bibfnamefont {G.}~\bibnamefont {Forte}}, \bibinfo {author} {\bibfnamefont {M.}~\bibnamefont {De~Falco}}, \bibinfo {author} {\bibfnamefont {G.~B.}\ \bibnamefont {Chirico}},\ and\ \bibinfo {author} {\bibfnamefont {A.}~\bibnamefont {Santo}},\ }\bibfield  {title} {\bibinfo {title} {New insights on rainfall triggering flow-like landslides and flash floods in {Campania} {(Southern Italy)}},\ }\href@noop {} {\bibfield  {journal} {\bibinfo  {journal} {Landslides}\ }\textbf {\bibinfo {volume} {18}},\ \bibinfo {pages} {2923} (\bibinfo {year} {2021})}\BibitemShut {NoStop}%
\bibitem [{\citenamefont {Guazzelli}(2024)}]{guazzelli2024rheology}%
  \BibitemOpen
  \bibfield  {author} {\bibinfo {author} {\bibfnamefont {{\'E}.}~\bibnamefont {Guazzelli}},\ }\bibfield  {title} {\bibinfo {title} {Rheology of dense granular suspensions across flow regimes},\ }\href@noop {} {\bibfield  {journal} {\bibinfo  {journal} {Physical Review Fluids}\ }\textbf {\bibinfo {volume} {9}},\ \bibinfo {pages} {090501} (\bibinfo {year} {2024})}\BibitemShut {NoStop}%
\bibitem [{\citenamefont {Larcher}\ and\ \citenamefont {Jenkins}(2019)}]{larcher2019influence}%
  \BibitemOpen
  \bibfield  {author} {\bibinfo {author} {\bibfnamefont {M.}~\bibnamefont {Larcher}}\ and\ \bibinfo {author} {\bibfnamefont {J.~T.}\ \bibnamefont {Jenkins}},\ }\bibfield  {title} {\bibinfo {title} {The influence of granular segregation on gravity-driven particle-fluid flows},\ }\href@noop {} {\bibfield  {journal} {\bibinfo  {journal} {Advances in Water Resources}\ }\textbf {\bibinfo {volume} {129}},\ \bibinfo {pages} {365} (\bibinfo {year} {2019})}\BibitemShut {NoStop}%
\bibitem [{\citenamefont {Balmforth}\ \emph {et~al.}(2007)\citenamefont {Balmforth}, \citenamefont {Craster}, \citenamefont {Perona}, \citenamefont {Rust},\ and\ \citenamefont {Sassi}}]{balmforth_viscoplastic_2007}%
  \BibitemOpen
  \bibfield  {author} {\bibinfo {author} {\bibfnamefont {N.~J.}\ \bibnamefont {Balmforth}}, \bibinfo {author} {\bibfnamefont {R.~V.}\ \bibnamefont {Craster}}, \bibinfo {author} {\bibfnamefont {P.}~\bibnamefont {Perona}}, \bibinfo {author} {\bibfnamefont {A.~C.}\ \bibnamefont {Rust}},\ and\ \bibinfo {author} {\bibfnamefont {R.}~\bibnamefont {Sassi}},\ }\bibfield  {title} {\bibinfo {title} {Viscoplastic dam breaks and the {Bostwick} consistometer},\ }\href {https://doi.org/10.1016/j.jnnfm.2006.06.005} {\bibfield  {journal} {\bibinfo  {journal} {Journal of Non-Newtonian Fluid Mechanics}\ }\bibinfo {series} {Viscoplastic fluids: {From} theory to application},\ \textbf {\bibinfo {volume} {142}},\ \bibinfo {pages} {63} (\bibinfo {year} {2007})}\BibitemShut {NoStop}%
\bibitem [{\citenamefont {Mouquet}\ \emph {et~al.}(2006)\citenamefont {Mouquet}, \citenamefont {Greffeuille},\ and\ \citenamefont {Treche}}]{mouquet_characterization_2006}%
  \BibitemOpen
  \bibfield  {author} {\bibinfo {author} {\bibfnamefont {C.}~\bibnamefont {Mouquet}}, \bibinfo {author} {\bibfnamefont {V.}~\bibnamefont {Greffeuille}},\ and\ \bibinfo {author} {\bibfnamefont {S.}~\bibnamefont {Treche}},\ }\bibfield  {title} {\bibinfo {title} {Characterization of the consistency of gruels consumed by infants in developing countries: assessment of the {Bostwick} consistometer and comparison with viscosity measurements and sensory perception},\ }\href {https://doi.org/10.1080/09637480600931618} {\bibfield  {journal} {\bibinfo  {journal} {International Journal of Food Sciences and Nutrition}\ }\textbf {\bibinfo {volume} {57}},\ \bibinfo {pages} {459} (\bibinfo {year} {2006})}\BibitemShut {NoStop}%
\bibitem [{\citenamefont {Tehrani}\ and\ \citenamefont {Ghandi}(2007)}]{tehrani_modification_2007}%
  \BibitemOpen
  \bibfield  {author} {\bibinfo {author} {\bibfnamefont {M.~M.}\ \bibnamefont {Tehrani}}\ and\ \bibinfo {author} {\bibfnamefont {A.}~\bibnamefont {Ghandi}},\ }\bibfield  {title} {\bibinfo {title} {Modification of {Bostwick} method to determine tomato concentrate consistency},\ }\href {https://doi.org/10.1016/j.jfoodeng.2006.01.093} {\bibfield  {journal} {\bibinfo  {journal} {Journal of Food Engineering}\ }\textbf {\bibinfo {volume} {79}},\ \bibinfo {pages} {1483} (\bibinfo {year} {2007})}\BibitemShut {NoStop}%
\bibitem [{\citenamefont {Wright}(1981)}]{WrightpatentEU}%
  \BibitemOpen
  \bibfield  {author} {\bibinfo {author} {\bibfnamefont {D.~C.}\ \bibnamefont {Wright}},\ }\href@noop {} {\bibinfo {title} {A spiral separator ({EU Patent} 0039139).}} (\bibinfo {year} {1981})\BibitemShut {NoStop}%
\bibitem [{\citenamefont {Dehaine}\ and\ \citenamefont {Filippov}(2016)}]{dehaine_modelling_2016}%
  \BibitemOpen
  \bibfield  {author} {\bibinfo {author} {\bibfnamefont {Q.}~\bibnamefont {Dehaine}}\ and\ \bibinfo {author} {\bibfnamefont {L.~O.}\ \bibnamefont {Filippov}},\ }\bibfield  {title} {\bibinfo {title} {Modelling heavy and gangue mineral size recovery curves using the spiral concentration of heavy minerals from kaolin residues},\ }\href {https://doi.org/10.1016/j.powtec.2016.02.005} {\bibfield  {journal} {\bibinfo  {journal} {Powder Technology}\ }\textbf {\bibinfo {volume} {292}},\ \bibinfo {pages} {331} (\bibinfo {year} {2016})}\BibitemShut {NoStop}%
\bibitem [{\citenamefont {Holland-Batt}\ and\ \citenamefont {Holtham}(1991)}]{holland-batt_particle_1991}%
  \BibitemOpen
  \bibfield  {author} {\bibinfo {author} {\bibfnamefont {A.~B.}\ \bibnamefont {Holland-Batt}}\ and\ \bibinfo {author} {\bibfnamefont {P.~N.}\ \bibnamefont {Holtham}},\ }\bibfield  {title} {\bibinfo {title} {Particle and fluid motion on spiral separators},\ }\href {https://doi.org/10.1016/0892-6875(91)90147-N} {\bibfield  {journal} {\bibinfo  {journal} {Minerals Engineering}\ }\textbf {\bibinfo {volume} {4}},\ \bibinfo {pages} {457} (\bibinfo {year} {1991})}\BibitemShut {NoStop}%
\bibitem [{\citenamefont {Rosato}\ \emph {et~al.}(1987)\citenamefont {Rosato}, \citenamefont {Strandburg}, \citenamefont {Prinz},\ and\ \citenamefont {Swendsen}}]{rosato_why_1987}%
  \BibitemOpen
  \bibfield  {author} {\bibinfo {author} {\bibfnamefont {A.}~\bibnamefont {Rosato}}, \bibinfo {author} {\bibfnamefont {K.~J.}\ \bibnamefont {Strandburg}}, \bibinfo {author} {\bibfnamefont {F.}~\bibnamefont {Prinz}},\ and\ \bibinfo {author} {\bibfnamefont {R.~H.}\ \bibnamefont {Swendsen}},\ }\bibfield  {title} {\bibinfo {title} {Why the {Brazil} nuts are on top: {Size} segregation of particulate matter by shaking},\ }\href {https://doi.org/10.1103/PhysRevLett.58.1038} {\bibfield  {journal} {\bibinfo  {journal} {Physical Review Letters}\ }\textbf {\bibinfo {volume} {58}},\ \bibinfo {pages} {1038} (\bibinfo {year} {1987})}\BibitemShut {NoStop}%
\bibitem [{\citenamefont {Jullien}\ and\ \citenamefont {Meakin}(1990)}]{jullien_mechanism_1990}%
  \BibitemOpen
  \bibfield  {author} {\bibinfo {author} {\bibfnamefont {R.}~\bibnamefont {Jullien}}\ and\ \bibinfo {author} {\bibfnamefont {P.}~\bibnamefont {Meakin}},\ }\bibfield  {title} {\bibinfo {title} {A mechanism for particle size segregation in three dimensions},\ }\href {https://doi.org/10.1038/344425a0} {\bibfield  {journal} {\bibinfo  {journal} {Nature}\ }\textbf {\bibinfo {volume} {344}},\ \bibinfo {pages} {425} (\bibinfo {year} {1990})}\BibitemShut {NoStop}%
\bibitem [{\citenamefont {Knight}\ \emph {et~al.}(1993)\citenamefont {Knight}, \citenamefont {Jaeger},\ and\ \citenamefont {Nagel}}]{knight_vibration-induced_1993}%
  \BibitemOpen
  \bibfield  {author} {\bibinfo {author} {\bibfnamefont {J.~B.}\ \bibnamefont {Knight}}, \bibinfo {author} {\bibfnamefont {H.~M.}\ \bibnamefont {Jaeger}},\ and\ \bibinfo {author} {\bibfnamefont {S.~R.}\ \bibnamefont {Nagel}},\ }\bibfield  {title} {\bibinfo {title} {Vibration-induced size separation in granular media: {The} convection connection},\ }\href {https://doi.org/10.1103/PhysRevLett.70.3728} {\bibfield  {journal} {\bibinfo  {journal} {Physical Review Letters}\ }\textbf {\bibinfo {volume} {70}},\ \bibinfo {pages} {3728} (\bibinfo {year} {1993})}\BibitemShut {NoStop}%
\bibitem [{\citenamefont {Möbius}\ \emph {et~al.}(2001)\citenamefont {Möbius}, \citenamefont {Lauderdale}, \citenamefont {Nagel},\ and\ \citenamefont {Jaeger}}]{mobius_size_2001}%
  \BibitemOpen
  \bibfield  {author} {\bibinfo {author} {\bibfnamefont {M.~E.}\ \bibnamefont {Möbius}}, \bibinfo {author} {\bibfnamefont {B.~E.}\ \bibnamefont {Lauderdale}}, \bibinfo {author} {\bibfnamefont {S.~R.}\ \bibnamefont {Nagel}},\ and\ \bibinfo {author} {\bibfnamefont {H.~M.}\ \bibnamefont {Jaeger}},\ }\bibfield  {title} {\bibinfo {title} {Size separation of granular particles},\ }\href {https://doi.org/10.1038/35104697} {\bibfield  {journal} {\bibinfo  {journal} {Nature}\ }\textbf {\bibinfo {volume} {414}},\ \bibinfo {pages} {270} (\bibinfo {year} {2001})}\BibitemShut {NoStop}%
\bibitem [{\citenamefont {Fan}\ and\ \citenamefont {Hill}(2011)}]{fan_phase_2011}%
  \BibitemOpen
  \bibfield  {author} {\bibinfo {author} {\bibfnamefont {Y.}~\bibnamefont {Fan}}\ and\ \bibinfo {author} {\bibfnamefont {K.~M.}\ \bibnamefont {Hill}},\ }\bibfield  {title} {\bibinfo {title} {Phase {Transitions} in {Shear}-{Induced} {Segregation} of {Granular} {Materials}},\ }\href {https://doi.org/10.1103/PhysRevLett.106.218301} {\bibfield  {journal} {\bibinfo  {journal} {Physical Review Letters}\ }\textbf {\bibinfo {volume} {106}},\ \bibinfo {pages} {218301} (\bibinfo {year} {2011})}\BibitemShut {NoStop}%
\bibitem [{\citenamefont {Knight}\ \emph {et~al.}(1996)\citenamefont {Knight}, \citenamefont {Ehrichs}, \citenamefont {Kuperman}, \citenamefont {Flint}, \citenamefont {Jaeger},\ and\ \citenamefont {Nagel}}]{knight_experimental_1996}%
  \BibitemOpen
  \bibfield  {author} {\bibinfo {author} {\bibfnamefont {J.~B.}\ \bibnamefont {Knight}}, \bibinfo {author} {\bibfnamefont {E.~E.}\ \bibnamefont {Ehrichs}}, \bibinfo {author} {\bibfnamefont {V.~Y.}\ \bibnamefont {Kuperman}}, \bibinfo {author} {\bibfnamefont {J.~K.}\ \bibnamefont {Flint}}, \bibinfo {author} {\bibfnamefont {H.~M.}\ \bibnamefont {Jaeger}},\ and\ \bibinfo {author} {\bibfnamefont {S.~R.}\ \bibnamefont {Nagel}},\ }\bibfield  {title} {\bibinfo {title} {Experimental study of granular convection},\ }\href {https://doi.org/10.1103/PhysRevE.54.5726} {\bibfield  {journal} {\bibinfo  {journal} {Physical Review E}\ }\textbf {\bibinfo {volume} {54}},\ \bibinfo {pages} {5726} (\bibinfo {year} {1996})}\BibitemShut {NoStop}%
\bibitem [{\citenamefont {Gajjar}\ \emph {et~al.}(2021)\citenamefont {Gajjar}, \citenamefont {Johnson}, \citenamefont {Carr}, \citenamefont {Chrispeels}, \citenamefont {Gray},\ and\ \citenamefont {Withers}}]{gajjar_size_2021}%
  \BibitemOpen
  \bibfield  {author} {\bibinfo {author} {\bibfnamefont {P.}~\bibnamefont {Gajjar}}, \bibinfo {author} {\bibfnamefont {C.~G.}\ \bibnamefont {Johnson}}, \bibinfo {author} {\bibfnamefont {J.}~\bibnamefont {Carr}}, \bibinfo {author} {\bibfnamefont {K.}~\bibnamefont {Chrispeels}}, \bibinfo {author} {\bibfnamefont {J.~M. N.~T.}\ \bibnamefont {Gray}},\ and\ \bibinfo {author} {\bibfnamefont {P.~J.}\ \bibnamefont {Withers}},\ }\bibfield  {title} {\bibinfo {title} {Size segregation of irregular granular materials captured by time-resolved {3D} imaging},\ }\href {https://doi.org/10.1038/s41598-021-87280-1} {\bibfield  {journal} {\bibinfo  {journal} {Scientific Reports}\ }\textbf {\bibinfo {volume} {11}},\ \bibinfo {pages} {8352} (\bibinfo {year} {2021})}\BibitemShut {NoStop}%
\bibitem [{\citenamefont {Burnett}\ \emph {et~al.}(2024)\citenamefont {Burnett}, \citenamefont {Luan}, \citenamefont {Bloom}, \citenamefont {Ding},\ and\ \citenamefont {Bertozzi}}]{burnett2024gfm}%
  \BibitemOpen
  \bibfield  {author} {\bibinfo {author} {\bibfnamefont {S.}~\bibnamefont {Burnett}}, \bibinfo {author} {\bibfnamefont {Q.}~\bibnamefont {Luan}}, \bibinfo {author} {\bibfnamefont {S.}~\bibnamefont {Bloom}}, \bibinfo {author} {\bibfnamefont {L.}~\bibnamefont {Ding}},\ and\ \bibinfo {author} {\bibfnamefont {A.~L.}\ \bibnamefont {Bertozzi}},\ }\href {https://doi.org/10.1103/APS.DFD.2024.GFM.V2691002} {\bibinfo {title} {Separation of bidisperse particles in viscous thin-film flow down an incline}},\ \bibinfo {howpublished} {Gallery of Fluid Motion, APS Division of Fluid Dynamics Meeting 2024} (\bibinfo {year} {2024})\BibitemShut {NoStop}%
\bibitem [{\citenamefont {Savage}\ and\ \citenamefont {Lun}(1988)}]{savage_particle_1988}%
  \BibitemOpen
  \bibfield  {author} {\bibinfo {author} {\bibfnamefont {S.~B.}\ \bibnamefont {Savage}}\ and\ \bibinfo {author} {\bibfnamefont {C.~K.~K.}\ \bibnamefont {Lun}},\ }\bibfield  {title} {\bibinfo {title} {Particle size segregation in inclined chute flow of dry cohesionless granular solids},\ }\href {https://doi.org/10.1017/S002211208800103X} {\bibfield  {journal} {\bibinfo  {journal} {Journal of Fluid Mechanics}\ }\textbf {\bibinfo {volume} {189}},\ \bibinfo {pages} {311} (\bibinfo {year} {1988})}\BibitemShut {NoStop}%
\bibitem [{\citenamefont {Neveu}\ \emph {et~al.}(2022)\citenamefont {Neveu}, \citenamefont {Larcher}, \citenamefont {Delannay}, \citenamefont {Jenkins},\ and\ \citenamefont {Valance}}]{neveu_particle_2022}%
  \BibitemOpen
  \bibfield  {author} {\bibinfo {author} {\bibfnamefont {A.}~\bibnamefont {Neveu}}, \bibinfo {author} {\bibfnamefont {M.}~\bibnamefont {Larcher}}, \bibinfo {author} {\bibfnamefont {R.}~\bibnamefont {Delannay}}, \bibinfo {author} {\bibfnamefont {J.~T.}\ \bibnamefont {Jenkins}},\ and\ \bibinfo {author} {\bibfnamefont {A.}~\bibnamefont {Valance}},\ }\bibfield  {title} {\bibinfo {title} {Particle segregation in inclined high-speed granular flows},\ }\href {https://doi.org/10.1017/jfm.2022.51} {\bibfield  {journal} {\bibinfo  {journal} {Journal of Fluid Mechanics}\ }\textbf {\bibinfo {volume} {935}},\ \bibinfo {pages} {A41} (\bibinfo {year} {2022})}\BibitemShut {NoStop}%
\bibitem [{\citenamefont {Johanson}(1978)}]{johanson_j_r_particle_1978}%
  \BibitemOpen
  \bibfield  {author} {\bibinfo {author} {\bibfnamefont {J.~R.}\ \bibnamefont {Johanson}},\ }\bibfield  {title} {\bibinfo {title} {Particle segregation and what to do about it},\ }\href@noop {} {\bibfield  {journal} {\bibinfo  {journal} {Chemical Engineering}\ }\textbf {\bibinfo {volume} {8}},\ \bibinfo {pages} {183} (\bibinfo {year} {1978})}\BibitemShut {NoStop}%
\bibitem [{\citenamefont {Murisic}\ \emph {et~al.}(2013)\citenamefont {Murisic}, \citenamefont {Pausader}, \citenamefont {Peschka},\ and\ \citenamefont {Bertozzi}}]{murisic_dynamics_2013}%
  \BibitemOpen
  \bibfield  {author} {\bibinfo {author} {\bibfnamefont {N.}~\bibnamefont {Murisic}}, \bibinfo {author} {\bibfnamefont {B.}~\bibnamefont {Pausader}}, \bibinfo {author} {\bibfnamefont {D.}~\bibnamefont {Peschka}},\ and\ \bibinfo {author} {\bibfnamefont {A.~L.}\ \bibnamefont {Bertozzi}},\ }\bibfield  {title} {\bibinfo {title} {Dynamics of particle settling and resuspension in viscous liquid films},\ }\href@noop {} {\bibfield  {journal} {\bibinfo  {journal} {Journal of Fluid Mechanics}\ }\textbf {\bibinfo {volume} {717}},\ \bibinfo {pages} {203} (\bibinfo {year} {2013})}\BibitemShut {NoStop}%
\bibitem [{\citenamefont {Huppert}(1982)}]{huppert1982flow}%
  \BibitemOpen
  \bibfield  {author} {\bibinfo {author} {\bibfnamefont {H.~E.}\ \bibnamefont {Huppert}},\ }\bibfield  {title} {\bibinfo {title} {Flow and instability of a viscous current down a slope},\ }\href@noop {} {\bibfield  {journal} {\bibinfo  {journal} {Nature}\ }\textbf {\bibinfo {volume} {300}},\ \bibinfo {pages} {427} (\bibinfo {year} {1982})}\BibitemShut {NoStop}%
\bibitem [{\citenamefont {Zhou}\ \emph {et~al.}(2005)\citenamefont {Zhou}, \citenamefont {Dupuy}, \citenamefont {Bertozzi},\ and\ \citenamefont {Hosoi}}]{zhou2005theory}%
  \BibitemOpen
  \bibfield  {author} {\bibinfo {author} {\bibfnamefont {J.}~\bibnamefont {Zhou}}, \bibinfo {author} {\bibfnamefont {B.}~\bibnamefont {Dupuy}}, \bibinfo {author} {\bibfnamefont {A.}~\bibnamefont {Bertozzi}},\ and\ \bibinfo {author} {\bibfnamefont {A.}~\bibnamefont {Hosoi}},\ }\bibfield  {title} {\bibinfo {title} {Theory for shock dynamics in particle-laden thin films},\ }\href@noop {} {\bibfield  {journal} {\bibinfo  {journal} {Physical Review Letters}\ }\textbf {\bibinfo {volume} {94}},\ \bibinfo {pages} {117803} (\bibinfo {year} {2005})}\BibitemShut {NoStop}%
\bibitem [{\citenamefont {Cook}(2008)}]{cook2008theory}%
  \BibitemOpen
  \bibfield  {author} {\bibinfo {author} {\bibfnamefont {B.~P.}\ \bibnamefont {Cook}},\ }\bibfield  {title} {\bibinfo {title} {Theory for particle settling and shear-induced migration in thin-film liquid flow},\ }\href@noop {} {\bibfield  {journal} {\bibinfo  {journal} {Physical Review E}\ }\textbf {\bibinfo {volume} {78}},\ \bibinfo {pages} {045303} (\bibinfo {year} {2008})}\BibitemShut {NoStop}%
\bibitem [{\citenamefont {Murisic}\ \emph {et~al.}(2011)\citenamefont {Murisic}, \citenamefont {Ho}, \citenamefont {Hu}, \citenamefont {Latterman}, \citenamefont {Koch}, \citenamefont {Lin}, \citenamefont {Mata},\ and\ \citenamefont {Bertozzi}}]{murisic_particle-laden_2011}%
  \BibitemOpen
  \bibfield  {author} {\bibinfo {author} {\bibfnamefont {N.}~\bibnamefont {Murisic}}, \bibinfo {author} {\bibfnamefont {J.}~\bibnamefont {Ho}}, \bibinfo {author} {\bibfnamefont {V.}~\bibnamefont {Hu}}, \bibinfo {author} {\bibfnamefont {P.}~\bibnamefont {Latterman}}, \bibinfo {author} {\bibfnamefont {T.}~\bibnamefont {Koch}}, \bibinfo {author} {\bibfnamefont {K.}~\bibnamefont {Lin}}, \bibinfo {author} {\bibfnamefont {M.}~\bibnamefont {Mata}},\ and\ \bibinfo {author} {\bibfnamefont {A.}~\bibnamefont {Bertozzi}},\ }\bibfield  {title} {\bibinfo {title} {Particle-laden viscous thin-film flows on an incline: {Experiments} compared with a theory based on shear-induced migration and particle settling},\ }\href@noop {} {\bibfield  {journal} {\bibinfo  {journal} {Physica D: Nonlinear Phenomena}\ }\textbf {\bibinfo {volume} {240}},\ \bibinfo {pages} {1661} (\bibinfo {year} {2011})}\BibitemShut {NoStop}%
\bibitem [{\citenamefont {Wong}\ and\ \citenamefont {Bertozzi}(2016)}]{wong_conservation_2016}%
  \BibitemOpen
  \bibfield  {author} {\bibinfo {author} {\bibfnamefont {J.~T.}\ \bibnamefont {Wong}}\ and\ \bibinfo {author} {\bibfnamefont {A.~L.}\ \bibnamefont {Bertozzi}},\ }\bibfield  {title} {\bibinfo {title} {A conservation law model for bidensity suspensions on an incline},\ }\href@noop {} {\bibfield  {journal} {\bibinfo  {journal} {Physica D: Nonlinear Phenomena}\ }\textbf {\bibinfo {volume} {330}},\ \bibinfo {pages} {47} (\bibinfo {year} {2016})}\BibitemShut {NoStop}%
\bibitem [{\citenamefont {Lee}\ \emph {et~al.}(2015)\citenamefont {Lee}, \citenamefont {Wong},\ and\ \citenamefont {Bertozzi}}]{lee_equilibrium_2015}%
  \BibitemOpen
  \bibfield  {author} {\bibinfo {author} {\bibfnamefont {S.}~\bibnamefont {Lee}}, \bibinfo {author} {\bibfnamefont {J.}~\bibnamefont {Wong}},\ and\ \bibinfo {author} {\bibfnamefont {A.~L.}\ \bibnamefont {Bertozzi}},\ }\bibfield  {title} {\bibinfo {title} {Equilibrium {Theory} of {Bidensity} {Particle}-{Laden} {Flows} on an {Incline}},\ }in\ \href@noop {} {\emph {\bibinfo {booktitle} {Mathematical {Modelling} and {Numerical} {Simulation} of {Oil} {Pollution} {Problems}}}},\ \bibinfo {editor} {edited by\ \bibinfo {editor} {\bibfnamefont {M.}~\bibnamefont {Ehrhardt}}}\ (\bibinfo  {publisher} {Springer International Publishing},\ \bibinfo {address} {Cham, Switzerland},\ \bibinfo {year} {2015})\ pp.\ \bibinfo {pages} {85--97}\BibitemShut {NoStop}%
\bibitem [{\citenamefont {Shauly}\ \emph {et~al.}(1998)\citenamefont {Shauly}, \citenamefont {Wachs},\ and\ \citenamefont {Nir}}]{shauly_shear-induced_1998}%
  \BibitemOpen
  \bibfield  {author} {\bibinfo {author} {\bibfnamefont {A.}~\bibnamefont {Shauly}}, \bibinfo {author} {\bibfnamefont {A.}~\bibnamefont {Wachs}},\ and\ \bibinfo {author} {\bibfnamefont {A.}~\bibnamefont {Nir}},\ }\bibfield  {title} {\bibinfo {title} {Shear-induced particle migration in a polydisperse concentrated suspension},\ }\href {https://doi.org/10.1122/1.550963} {\bibfield  {journal} {\bibinfo  {journal} {Journal of Rheology}\ }\textbf {\bibinfo {volume} {42}},\ \bibinfo {pages} {1329} (\bibinfo {year} {1998})}\BibitemShut {NoStop}%
\bibitem [{\citenamefont {Kanehl}\ and\ \citenamefont {Stark}(2015)}]{kanehl_hydrodynamic_2015}%
  \BibitemOpen
  \bibfield  {author} {\bibinfo {author} {\bibfnamefont {P.}~\bibnamefont {Kanehl}}\ and\ \bibinfo {author} {\bibfnamefont {H.}~\bibnamefont {Stark}},\ }\bibfield  {title} {\bibinfo {title} {Hydrodynamic segregation in a bidisperse colloidal suspension in microchannel flow: {A} theoretical study},\ }\href {https://doi.org/10.1063/1.4921800} {\bibfield  {journal} {\bibinfo  {journal} {J. Chem. Phys.}\ }\textbf {\bibinfo {volume} {142}},\ \bibinfo {pages} {214901} (\bibinfo {year} {2015})}\BibitemShut {NoStop}%
\bibitem [{\citenamefont {Howard}\ \emph {et~al.}(2022)\citenamefont {Howard}, \citenamefont {Maxey},\ and\ \citenamefont {Gallier}}]{howard_bidisperse_2022}%
  \BibitemOpen
  \bibfield  {author} {\bibinfo {author} {\bibfnamefont {A.~A.}\ \bibnamefont {Howard}}, \bibinfo {author} {\bibfnamefont {M.~R.}\ \bibnamefont {Maxey}},\ and\ \bibinfo {author} {\bibfnamefont {S.}~\bibnamefont {Gallier}},\ }\bibfield  {title} {\bibinfo {title} {Bidisperse suspension balance model},\ }\href {https://doi.org/10.1103/PhysRevFluids.7.124301} {\bibfield  {journal} {\bibinfo  {journal} {Physical Review Fluids}\ }\textbf {\bibinfo {volume} {7}},\ \bibinfo {pages} {124301} (\bibinfo {year} {2022})}\BibitemShut {NoStop}%
\bibitem [{\citenamefont {Thornton}\ \emph {et~al.}(2006)\citenamefont {Thornton}, \citenamefont {Gray},\ and\ \citenamefont {Hogg}}]{thornton_three-phase_2006}%
  \BibitemOpen
  \bibfield  {author} {\bibinfo {author} {\bibfnamefont {A.~R.}\ \bibnamefont {Thornton}}, \bibinfo {author} {\bibfnamefont {J.~M. N.~T.}\ \bibnamefont {Gray}},\ and\ \bibinfo {author} {\bibfnamefont {A.~J.}\ \bibnamefont {Hogg}},\ }\bibfield  {title} {\bibinfo {title} {A three-phase mixture theory for particle size segregation in shallow granular free-surface flows},\ }\href {https://doi.org/10.1017/S0022112005007676} {\bibfield  {journal} {\bibinfo  {journal} {Journal of Fluid Mechanics}\ }\textbf {\bibinfo {volume} {550}},\ \bibinfo {pages} {1} (\bibinfo {year} {2006})}\BibitemShut {NoStop}%
\bibitem [{\citenamefont {Ding}\ \emph {et~al.}(2025)\citenamefont {Ding}, \citenamefont {Burnett},\ and\ \citenamefont {Bertozzi}}]{ding_equilibrium_2025}%
  \BibitemOpen
  \bibfield  {author} {\bibinfo {author} {\bibfnamefont {L.}~\bibnamefont {Ding}}, \bibinfo {author} {\bibfnamefont {S.~C.}\ \bibnamefont {Burnett}},\ and\ \bibinfo {author} {\bibfnamefont {A.~L.}\ \bibnamefont {Bertozzi}},\ }\bibfield  {title} {\bibinfo {title} {Equilibrium theory of bidensity particle-laden suspensions in thin-film flow down a spiral separator},\ }\href {https://doi.org/10.1063/5.0246347} {\bibfield  {journal} {\bibinfo  {journal} {Physics of Fluids}\ }\textbf {\bibinfo {volume} {37}},\ \bibinfo {pages} {023397} (\bibinfo {year} {2025})}\BibitemShut {NoStop}%
\bibitem [{\citenamefont {Krieger}\ and\ \citenamefont {Dougherty}(1959)}]{krieger_mechanism_1959}%
  \BibitemOpen
  \bibfield  {author} {\bibinfo {author} {\bibfnamefont {I.~M.}\ \bibnamefont {Krieger}}\ and\ \bibinfo {author} {\bibfnamefont {T.~J.}\ \bibnamefont {Dougherty}},\ }\bibfield  {title} {\bibinfo {title} {A mechanism for non‐newtonian flow in suspensions of rigid spheres},\ }\href {https://doi.org/10.1122/1.548848} {\bibfield  {journal} {\bibinfo  {journal} {Transactions of The Society of Rheology}\ }\textbf {\bibinfo {volume} {3}},\ \bibinfo {pages} {137} (\bibinfo {year} {1959})}\BibitemShut {NoStop}%
\bibitem [{\citenamefont {Ward}\ \emph {et~al.}(2009)\citenamefont {Ward}, \citenamefont {Wey}, \citenamefont {Glidden}, \citenamefont {Hosoi},\ and\ \citenamefont {Bertozzi}}]{ward_experimental_2009}%
  \BibitemOpen
  \bibfield  {author} {\bibinfo {author} {\bibfnamefont {T.}~\bibnamefont {Ward}}, \bibinfo {author} {\bibfnamefont {C.}~\bibnamefont {Wey}}, \bibinfo {author} {\bibfnamefont {R.}~\bibnamefont {Glidden}}, \bibinfo {author} {\bibfnamefont {A.~E.}\ \bibnamefont {Hosoi}},\ and\ \bibinfo {author} {\bibfnamefont {A.~L.}\ \bibnamefont {Bertozzi}},\ }\bibfield  {title} {\bibinfo {title} {Experimental study of gravitation effects in the flow of a particle-laden thin film on an inclined plane},\ }\href {https://doi.org/10.1063/1.3208076} {\bibfield  {journal} {\bibinfo  {journal} {Physics of Fluids}\ }\textbf {\bibinfo {volume} {21}},\ \bibinfo {pages} {083305} (\bibinfo {year} {2009})}\BibitemShut {NoStop}%
\bibitem [{\citenamefont {Leighton}\ and\ \citenamefont {Acrivos}(1987{\natexlab{a}})}]{leighton1987shear}%
  \BibitemOpen
  \bibfield  {author} {\bibinfo {author} {\bibfnamefont {D.}~\bibnamefont {Leighton}}\ and\ \bibinfo {author} {\bibfnamefont {A.}~\bibnamefont {Acrivos}},\ }\bibfield  {title} {\bibinfo {title} {The shear-induced migration of particles in concentrated suspensions},\ }\href@noop {} {\bibfield  {journal} {\bibinfo  {journal} {Journal of Fluid Mechanics}\ }\textbf {\bibinfo {volume} {181}},\ \bibinfo {pages} {415} (\bibinfo {year} {1987}{\natexlab{a}})}\BibitemShut {NoStop}%
\bibitem [{\citenamefont {Phillips}\ \emph {et~al.}(1992)\citenamefont {Phillips}, \citenamefont {Armstrong}, \citenamefont {Brown}, \citenamefont {Graham},\ and\ \citenamefont {Abbott}}]{phillips1992constitutive}%
  \BibitemOpen
  \bibfield  {author} {\bibinfo {author} {\bibfnamefont {R.~J.}\ \bibnamefont {Phillips}}, \bibinfo {author} {\bibfnamefont {R.~C.}\ \bibnamefont {Armstrong}}, \bibinfo {author} {\bibfnamefont {R.~A.}\ \bibnamefont {Brown}}, \bibinfo {author} {\bibfnamefont {A.~L.}\ \bibnamefont {Graham}},\ and\ \bibinfo {author} {\bibfnamefont {J.~R.}\ \bibnamefont {Abbott}},\ }\bibfield  {title} {\bibinfo {title} {A constitutive equation for concentrated suspensions that accounts for shear-induced particle migration},\ }\href@noop {} {\bibfield  {journal} {\bibinfo  {journal} {Physics of Fluids A: Fluid Dynamics}\ }\textbf {\bibinfo {volume} {4}},\ \bibinfo {pages} {30} (\bibinfo {year} {1992})}\BibitemShut {NoStop}%
\bibitem [{\citenamefont {Tripathi}\ and\ \citenamefont {Acrivos}(1999)}]{tripathi_viscous_1999}%
  \BibitemOpen
  \bibfield  {author} {\bibinfo {author} {\bibfnamefont {A.}~\bibnamefont {Tripathi}}\ and\ \bibinfo {author} {\bibfnamefont {A.}~\bibnamefont {Acrivos}},\ }\bibfield  {title} {\bibinfo {title} {Viscous resuspension in a bidensity suspension},\ }\href {https://doi.org/10.1016/S0301-9322(98)00036-6} {\bibfield  {journal} {\bibinfo  {journal} {International Journal of Multiphase Flow}\ }\textbf {\bibinfo {volume} {25}},\ \bibinfo {pages} {1} (\bibinfo {year} {1999})}\BibitemShut {NoStop}%
\bibitem [{\citenamefont {Leighton}\ and\ \citenamefont {Acrivos}(1987{\natexlab{b}})}]{leighton_measurement_1987}%
  \BibitemOpen
  \bibfield  {author} {\bibinfo {author} {\bibfnamefont {D.}~\bibnamefont {Leighton}}\ and\ \bibinfo {author} {\bibfnamefont {A.}~\bibnamefont {Acrivos}},\ }\bibfield  {title} {\bibinfo {title} {Measurement of shear-induced self-diffusion in concentrated suspensions of spheres},\ }\href {https://doi.org/10.1017/S0022112087000880} {\bibfield  {journal} {\bibinfo  {journal} {Journal of Fluid Mechanics}\ }\textbf {\bibinfo {volume} {177}},\ \bibinfo {pages} {109} (\bibinfo {year} {1987}{\natexlab{b}})}\BibitemShut {NoStop}%
\bibitem [{\citenamefont {Sierou}\ and\ \citenamefont {Brady}(2004)}]{sierou_shear-induced_2004}%
  \BibitemOpen
  \bibfield  {author} {\bibinfo {author} {\bibfnamefont {A.}~\bibnamefont {Sierou}}\ and\ \bibinfo {author} {\bibfnamefont {J.~F.}\ \bibnamefont {Brady}},\ }\bibfield  {title} {\bibinfo {title} {Shear-induced self-diffusion in non-colloidal suspensions},\ }\href {https://doi.org/10.1017/S0022112004008651} {\bibfield  {journal} {\bibinfo  {journal} {Journal of Fluid Mechanics}\ }\textbf {\bibinfo {volume} {506}},\ \bibinfo {pages} {285} (\bibinfo {year} {2004})}\BibitemShut {NoStop}%
\bibitem [{\citenamefont {Lax}(1973)}]{lax_1_1973}%
  \BibitemOpen
  \bibfield  {author} {\bibinfo {author} {\bibfnamefont {P.~D.}\ \bibnamefont {Lax}},\ }\bibfield  {title} {\bibinfo {title} {1. {Hyperbolic} {Systems} of {Conservation} {Laws} and the {Mathematical} {Theory} of {Shock} {Waves}},\ }in\ \href {https://doi.org/10.1137/1.9781611970562.ch1} {\emph {\bibinfo {booktitle} {Hyperbolic {Systems} of {Conservation} {Laws} and the {Mathematical} {Theory} of {Shock} {Waves}}}},\ \bibinfo {series and number} {{CBMS}-{NSF} {Regional} {Conference} {Series} in {Applied} {Mathematics}}\ (\bibinfo  {publisher} {Society for Industrial and Applied Mathematics},\ \bibinfo {year} {1973})\ pp.\ \bibinfo {pages} {1--48}\BibitemShut {NoStop}%
\bibitem [{\citenamefont {Wang}\ \emph {et~al.}(2015)\citenamefont {Wang}, \citenamefont {Mavromoustaki}, \citenamefont {Bertozzi}, \citenamefont {Urdaneta},\ and\ \citenamefont {Huang}}]{wang_rarefaction-singular_2015}%
  \BibitemOpen
  \bibfield  {author} {\bibinfo {author} {\bibfnamefont {L.}~\bibnamefont {Wang}}, \bibinfo {author} {\bibfnamefont {A.}~\bibnamefont {Mavromoustaki}}, \bibinfo {author} {\bibfnamefont {A.~L.}\ \bibnamefont {Bertozzi}}, \bibinfo {author} {\bibfnamefont {G.}~\bibnamefont {Urdaneta}},\ and\ \bibinfo {author} {\bibfnamefont {K.}~\bibnamefont {Huang}},\ }\bibfield  {title} {\bibinfo {title} {Rarefaction-singular shock dynamics for conserved volume gravity driven particle-laden thin film},\ }\href {https://doi.org/10.1063/1.4913851} {\bibfield  {journal} {\bibinfo  {journal} {Physics of Fluids}\ }\textbf {\bibinfo {volume} {27}},\ \bibinfo {pages} {033301} (\bibinfo {year} {2015})}\BibitemShut {NoStop}%
\bibitem [{\citenamefont {Wang}\ and\ \citenamefont {Bertozzi}(2014)}]{wang2014shock}%
  \BibitemOpen
  \bibfield  {author} {\bibinfo {author} {\bibfnamefont {L.}~\bibnamefont {Wang}}\ and\ \bibinfo {author} {\bibfnamefont {A.~L.}\ \bibnamefont {Bertozzi}},\ }\bibfield  {title} {\bibinfo {title} {Shock solutions for high concentration particle-laden thin films},\ }\href@noop {} {\bibfield  {journal} {\bibinfo  {journal} {SIAM Journal on Applied Mathematics}\ }\textbf {\bibinfo {volume} {74}},\ \bibinfo {pages} {322} (\bibinfo {year} {2014})}\BibitemShut {NoStop}%
\end{thebibliography}%

\end{document}